\documentclass[lettersize,journal]{IEEEtran}
\usepackage{amsmath,amsfonts}
\usepackage{booktabs}
\usepackage{cite}
\usepackage{array}
\usepackage{textcomp}
\usepackage{amsthm}
\usepackage[english]{babel}
\newtheorem{theorem}{Theorem}
\newtheorem{lemma}{Lemma}

\newtheorem{assumption}{Assumption}
\newtheorem{proposition}{Proposition}
\usepackage{stfloats}
\usepackage{graphicx}
\usepackage{xcolor}
\definecolor{my1}{RGB}{188, 100, 100}
\definecolor{my2}{RGB}{188, 188, 100}
\usepackage[caption=false,font=footnotesize]{subfig}

\usepackage{url}
\usepackage{verbatim}
\usepackage{graphicx}
\usepackage{amssymb}
\usepackage{algorithm2e}
\usepackage{amsmath}
\allowdisplaybreaks
\RestyleAlgo{ruled}
\def\BibTeX{{\rm B\kern-.05em{\sc i\kern-.025em b}\kern-.08em
    T\kern-.1667em\lower.7ex\hbox{E}\kern-.125emX}}
\usepackage{balance}
\usepackage{xcolor}
\usepackage{algpseudocode}
\usepackage{ifthen}
\usepackage{xcolor} % For defining custom colors
\usepackage{soul}   % For highlighting
 
\usepackage{multirow} % Add this line in your preamble
\definecolor{lightblue}{RGB}{173,216,230} % Define custom color
\sethlcolor{lightblue}                    % Set highlight color

\begin{document}
\title{Biased Compression in Gradient Coding for Distributed Learning}
\author{Chengxi Li,~\IEEEmembership{Member,~IEEE,} Ming Xiao,~\IEEEmembership{Senior Member,~IEEE,} and Mikael Skoglund,~\IEEEmembership{Fellow, IEEE}
\thanks{C. Li, M. Xiao and M. Skoglund are with the Division of Information Science and Engineering, School of Electrical Engineering and Computer Science, KTH Royal Institute of Technology, 10044 Stockholm, Sweden. (e-mail: chengxli@kth.se; mingx@kth.se; skoglund@kth.se). Corresponding author: Chengxi Li.}
}

\markboth{Journal of \LaTeX\ Class Files}%
{}

\maketitle

\begin{abstract}
Communication bottlenecks and the presence of stragglers pose significant challenges in distributed learning (DL). To deal with these challenges, recent advances leverage unbiased compression functions and gradient coding. However, the significant benefits of biased compression remain largely unexplored. To close this gap, we propose \textbf{Co}mpressed Gradient \textbf{Co}ding with \textbf{E}rror \textbf{F}eedback (COCO-EF), a novel DL method that combines gradient coding with biased compression to mitigate straggler effects and reduce communication costs. In each iteration, non-straggler devices encode local gradients from redundantly allocated training data, incorporate prior compression errors, and compress the results using biased compression functions before transmission. The server aggregates these compressed messages from the non-stragglers to approximate the global gradient for model updates. We provide rigorous theoretical convergence guarantees for COCO-EF and validate its superior learning performance over baseline methods through empirical evaluations. As far as we know, we are among the first to rigorously demonstrate that biased compression has substantial benefits in DL, when gradient coding is employed to cope with stragglers. 
\end{abstract}

\section{Introduction}
\label{introduction} 
Distributed learning (DL) is emerging as an important paradigm for training machine learning models in a distributed fashion, which increases computation efficiency by using computational resources on edge devices simultaneously \cite{verbraeken2020survey,wang2020convergence,makkuvalaser2024}. Under the DL paradigm, multiple devices coordinate with a central server. Prior to training, the training data are divided into subsets and allocated to the devices so that each device holds one or a few of the subsets \cite{jahani2021optimal}. Afterward, the training is implemented over multiple iterations. In each training iteration, the server transmits the current global model to the devices. With this global model, the devices compute local gradients with the local subsets and transmit them to the server. After aggregating the gradients received from the devices, the server updates the global model \cite{mcmahan2017communication,cao2023communication}. Despite the benefits, there are two practical constraints encountered in DL problems, which are introduced as follows. 

The first one is the \textbf{communication bottleneck} caused by the transmission of a large number of high-dimensional and real-valued vectors from the devices to the server, given limited communication resources \cite{gorbunov2020unified,tang2020communication,qian2022distributed,li2024decentralized,wang2018atomo}. Existing works have aimed to address the communication bottleneck by reducing the communication overhead in each iteration, which is attained by using various compression functions to compress the gradients into some specific vectors that can be represented by fewer bits.
% For example, some propose to quantize the local gradients before transmission to the server, which means fewer bits are required to send each element in the gradients. Others impose sparsity on the gradient vectors before communicating to the server, where only a small fraction of elements are non-zero.
These compression functions can be classified into two categories: unbiased compression functions and biased compression functions, depending on whether the output of the compression function is an unbiased estimate or a biased estimate of the vector before compression. For example, stochastic quantization \cite{alistarh2017qsgd,safaryan2021stochastic} and amplified rand-$K$ sparsification \cite{wangni2018gradient} belong to unbiased compression functions, while the sign-bit quantization \cite{bernstein2018signsgd} and top‑$K$ sparsification \cite{alistarh2018convergence,stich2018sparsified,sahu2021rethinking} belong to biased compression functions. Compared to biased compression functions, unbiased compression functions are more extensively studied and understood, and they guarantee convergence to the stationary point in many DL tasks \cite{alistarh2017qsgd,wangni2018gradient,li2020unified,safaryan2021stochastic,condat2024locodl,he2024unbiased}.  However, biased compression has certain advantages over unbiased compression when both are evaluated in terms of their average capacity to retain information in the gradients \cite{beznosikov2023biased}. This is because smaller approximation error is introduced during the compression process when biased compression is adopted. Furthermore, another drawback of unbiased compression in DL is that, when stochastic noise contaminates the gradients, unbiased compression suffers from a linear deterioration in convergence performance\footnote{{Note that this linear deterioration in convergence performance still occurs even when unbiased compression is combined with gradient-difference compression to compensate for the compression error \cite{stich2020communication}.}} \cite{stich2020communication}.
Consequently, biased compression functions have been adopted in practice to achieve superior empirical performance compared to unbiased compression functions, with a remedy known as the error feedback mechanism used to ensure convergence \cite{li2023analysis,horvath2020better}. The error feedback mechanism was originally used in \cite{seide20141}, where 
it was empirically shown that when training deep neural networks, the loss of training accuracy caused by quantization error can be compensated by using error feedback. In \cite{karimireddy2019error}, error feedback was analyzed for one-device case and later extended to a distributed version with multiple devices in \cite{beznosikov2023biased}. With this mechanism, in each iteration, each device transmits the compressed local gradient to the server and stores the compression error locally, which can be used in the next iteration to modify the local gradient in order to correct the bias in compression and guarantee the correct direction of model update \cite{beznosikov2023biased}.

Second, in practice, significant delays may occur on some devices, known as \textbf{stragglers}, when computing and communicating to the server \cite{tandon2017gradient,buyukates2020timely}. This issue becomes obvious when a large number of cheap devices are deployed to reduce the overall system cost \cite{hard2024learning}. In this case, it is more desirable to wait only for messages from the non-stragglers, the devices that respond promptly, to ensure efficient training. As a result, in each iteration, only the non-stragglers participate, which degrades the learning performance due to missing information from the stragglers. To deal with the negative impact of the stragglers in DL, the gradient coding techniques have been developed. For example, in \cite{tandon2017gradient}, an exact gradient coding strategy is proposed, wherein the training data are divided into subsets, replicated, and allocated redundantly to devices in a carefully planned way. In each training iteration, each non-straggler device transmits a coded vector to the server, which is a single linear combination of its local gradients computed from the local subsets. By carefully designing the encoding strategies at the devices and the decoding strategies at the server, the true global gradient can be recovered exactly without knowing which devices are stragglers. To enhance flexibility in coping with variability in the number of stragglers over iterations, stochastic gradient coding is proposed in \cite{bitar2020stochastic}. In this approach, the server allocates the training subsets to the devices with a low level of redundancy according to a pairwise balanced scheme. The devices compute local gradients on their subsets and send a linear combination of these gradients. The server aggregates messages from the non-stragglers to approximately recover the global gradient and updates the model. Although these methods can effectively mitigate the negative impact of stragglers, they suffer from a significant communication bottleneck. This is because the non-straggler devices all transmit high-dimensional and real-valued dense vectors to the server in each iteration, which induces a heavy communication overhead \cite{chengxi20241-BitGC}. 

There is very little work that aims to address both the communication bottleneck and the issue of stragglers simultaneously. One recent effort tackles this challenge. In \cite{chengxi20241-BitGC}, a DL method based on 1-bit gradient coding is proposed, where locally computed gradients are encoded and compressed into 1-bit vectors via unbiased stochastic quantization. These vectors are then transmitted from the non-straggler devices to the server. Based on the gradient coding technique under data allocation redundancy, the straggler effects can be mitigated effectively in \cite{chengxi20241-BitGC}. Meanwhile, by transmitting 1-bit vectors instead of the real-valued ones, the communication overhead can be reduced compared to the prior works. However, the potential advantages of biased compression functions are not exploited in \cite{chengxi20241-BitGC}, which could further enhance learning performance under the constraints of communication and stragglers. 

In this work, to exploit the advantages of gradient coding and biased compression functions in DL under the communication bottleneck with stragglers, we propose a new DL method, compressed gradient coding with error feedback (COCO-EF). COCO-EF is a meta-algorithm, which can be applied with any biased compressed functions. In COCO-EF, before the training starts, the training data are divided into subsets and allocated to the devices redundantly in a pairwise balanced scheme, motivated by stochastic gradient coding in \cite{bitar2020stochastic}. In each training iteration, the server broadcasts the current global model to all devices, and the devices compute the local gradients based on the local subsets of the training data. After that, each non-straggler device encodes the local gradients into a single vector, and incorporates the compression error stored from the previous iterations into this coded vector. By using a specific type of biased compression function, the coded vector incorporated with error is compressed and transmitted to the server to reduce the communication overhead. After receiving from the non-stragglers, based on the redundancy of training data allocation, the server approximately reconstruct the global gradient for model update. We analyze the convergence performance of COCO-EF theoretically for smooth loss functions, and show that it can attain a convergence rate of \(O\left( {\frac{1}{{\sqrt {T } }}} \right)\). Finally, we provide various numerical results to demonstrate the superiority of the proposed method compared to the baselines and verify our theoretical findings.   

Compared with existing works related to this topic, the differences between our approach and prior studies are summarized as follows:
\begin{itemize}
    \item Both our method and \cite{bitar2020stochastic} adopt a stochastic gradient coding scheme to deal with the stragglers. However, in \cite{bitar2020stochastic}, local gradients are encoded and transmitted to the server without any compression, leading to a substantial communication burden. In contrast, in our proposed method, compressed messages are transmitted by the devices to reduce the communication overhead. 
    \item Relative to \cite{chengxi20241-BitGC}, both our work and \cite{chengxi20241-BitGC} compress coded vectors under the same stochastic gradient coding scheme before transmission to the server, in order to reduce the communication overhead and to address the stragglers. However, \cite{chengxi20241-BitGC} employs a specific type of unbiased compression function. In contrast, our approach leverages biased compression functions to achieve further improvements in communication reduction in the presence of stragglers.
    \item  Similar to \cite{karimireddy2019error,beznosikov2023biased}, our method also exploits the advantages of biased compression functions with the error feedback mechanism to reduce the communication overhead for DL tasks. However, \cite{karimireddy2019error,beznosikov2023biased} do not consider stragglers and incorporate the compression error directly into the local gradients. With this operation, their learning performance degrades significantly in the presence of stragglers due to the absence of information from these stragglers. Considering this, how to exploit the potential of biased compression functions with the error feedback mechanism to reduce communication overhead in the presence of stragglers remains an open problem.
We address these challenges thoroughly in this work. In the presence of potential stragglers, our method combines gradient coding with biased compression and error feedback to handle stragglers while simultaneously reducing communication overhead. In this scheme, the compression error is incorporated into coded gradients on the devices, which is a very different operation compared with \cite{karimireddy2019error,beznosikov2023biased}, resulting in a substantially different analysis and empirically superior learning performance. 
\end{itemize}

Based on that, our contributions are stated as follows:
\begin{itemize}
    \item We propose a new DL method, COCO-EF, to improve the learning performance of previous works under the challenges of the communication bottleneck and stragglers. On one hand, leveraging redundant training data allocation and the gradient coding scheme, COCO-EF compensates for missing information from stragglers by utilizing messages from non-stragglers, thus mitigating the negative impact of stragglers. On the other hand, by encoding and compressing local gradients using biased compression functions before transmission to the server, the communication overhead is reduced significantly. Specifically, to offset information loss and the bias in compression, an error feedback mechanism is incorporated into COCO-EF, ensuring satisfactory learning performance despite the reduced communication overhead. 
    \item We rigorously characterize the convergence performance of the proposed method for smooth loss functions with constant learning rates, and show that it achieves a convergence rate of \( O\left( \frac{1}{\sqrt{T}} \right) \), which improves upon the \( O\left( \frac{1}{T^{1/4}} \right) \) rate established in~\cite{chengxi20241-BitGC}.
    \item We present extensive numerical results to demonstrate the superiority of the proposed method over the baselines and to validate our theoretical findings, which show that COCO-EF attains significantly better learning performance under the same communication overhead in the presence of stragglers.  
\end{itemize}

The rest of this paper is organized as follows. In Section~\ref{problem model}, the problem model is introduced. In Section~\ref{our method}, we propose our method. In Section~\ref{performance analysis}, we analyze the performance of the proposed method. Numerical results are provided and discussed in Section~\ref{simulations}. The conclusions are presented in Section~\ref{conclusions}. We provide all missing proofs in the appendices. 

\section{Problem Model}
\label{problem model}
The considered problem is formally formulated as follows. Suppose there are $N$ devices and a central server, which collaborate to train a model with a training dataset \({\mathcal{W}}\) composed of $M$ subsets: \(\mathcal{W} = \left\{ {{\mathcal{W}_1},...,{\mathcal{W}_M}} \right\}\). This is equivalent to solving the following optimization problem \cite{chengxi20241-BitGC,bitar2020stochastic}:  
\begin{align}
    \label{basic problem}
   \arg {\min _{{\boldsymbol{\theta }} \in {\mathbb{R}^D}}}F\left( {\boldsymbol{\theta }} \right), 
\end{align}
 where $\boldsymbol{\theta}$ is the model parameter vector, and $F\left( {\boldsymbol{\theta}} \right)$ represents the overall training loss expressed as 
\begin{align}
\label{overall loss}
    F\left( {\boldsymbol{\theta }} \right) = \sum\limits_{k = 1}^M {{f_k}\left( {\boldsymbol{\theta }} \right)}.
\end{align}
In (\ref{overall loss}), $f_k$ is the training loss associated with subset ${\mathcal{W}_k}$. 

Typically, under the DL framework, before the training starts, the subsets of \({\mathcal{W}}\) are allocated to the devices \cite{ye2018communication,bitar2020stochastic,chengxi20241-BitGC}.  Here, the data allocation is determined by the system designer and is therefore known to the server and all devices.  Afterward, the training is implemented by multiple iterations. In each iteration, the server broadcasts the global model to all devices, and each device computes the local gradients based on its local training data and transmits them to the server. After aggregating the local gradients from the devices, the server computes the global gradient and updates the global model accordingly, which ends the current iteration \cite{mcmahan2017communication,cao2023communication}. 

In practice, some stragglers among the devices may suffer from significant delays during local computations or communication with the server due to various incidents. In such cases, only the non-stragglers participate in the training \cite{tandon2017gradient}.  In the absence of detailed prior information about stragglers, it is reasonable to assume that, in each iteration, any device can become a straggler with probability \(p\), and the straggler behavior is independent both across iterations and among different devices \cite{bitar2020stochastic,li2024decentralized,chengxi20241-BitGC,adikari2020decentralized}.  It is intuitive that the missing information from the stragglers can degrade the learning performance compared to cases where no stragglers are present. Mitigating the negative impact of stragglers is necessary for DL in various applications. Additionally, considering that in each iteration multiple non-straggler devices transmit high-dimensional and real-valued vectors to the server, and that communication resources in the DL system are limited, it is crucial to reduce the communication overhead while maintaining satisfactory learning performance \cite{chengxi20241-BitGC}.

\section{COCO-EF: Compressed Gradient Coding with Error Feedback}
\label{our method}

In this section, we propose a new method to cope with the problem introduced in Section~\ref{problem model}. 

Before the training, the subsets in \(\mathcal{W} = \left\{ \mathcal{W}_1, \dots, \mathcal{W}_M \right\}\) are allocated to the devices in a pairwise balanced scheme, as done in the stochastic gradient coding scheme \cite{bitar2020stochastic}.  This scheme is adopted since it can be easily approximated by a completely random distribution of the training subsets in practice.  In this case, subset $\mathcal{W}_k$ is allocated to $d_k$ devices, \(\forall k\), and the number of devices that hold both $\mathcal{W}_{k_1}$ and $\mathcal{W}_{k_2}$ is \(\frac{{d_{k_1}}{d_{k_2}}}{N}\) for \(k_1 \neq k_2\). Based on that, we can use a matrix \(\mathbf{S}\) to represent the allocation of the training subsets among the devices, with \(s(i,k)\) being the \((i,k)\)-th element. In \(\mathbf{S}\), \(s(i,k) = 1\) indicates \(\mathcal{W}_k\) is allocated to device \(i\), while \(s(i,k) = 0\) implies the opposite. Each device $i$ initializes its local error vector as ${\mathbf{e}}_i^0 = 0$, $\forall i$, which will be used to store the compression error during the training process. 

During the training, in each iteration $t$, the server broadcasts the current global model $\boldsymbol{\theta}^t$ to all devices. After that, each non-straggler device $i$ computes the local gradients corresponding to its local subsets and obtains $\left\{ {\nabla {f_k}\left( {{{\boldsymbol{\theta}}^t}} \right)|k\in \left\{ {1,...,M} \right\},s\left( {i,k} \right) \ne 0} \right\}$. With these local gradients, device $i$ encodes them as 
\begin{align}
    \label{encoding}
    {\mathbf{g}}_i^t = \sum\limits_{k \in \mathcal{S}_i} {\frac{1}{{{d_k}{\left(1-p\right)}}}} \nabla {f_k}\left( {{{\boldsymbol{\theta}}^t}} \right),
\end{align}
where ${\mathcal{S}_i} \triangleq \left\{ {\left. k \right|s\left( {i,k} \right) \ne 0} \right\}$. After that, motivated by the error feedback mechanism\footnote{{In addition to error feedback, other techniques can be used to mitigate the negative impact of compression errors, such as compression of gradient differences \cite{stich2020communication}. However, since error feedback is specifically designed for biased compression \cite{karimireddy2019error,beznosikov2023biased}, whereas compression of gradient differences is designed for unbiased compression, we adopt error feedback in our method.}} in \cite{karimireddy2019error}, device $i$ incorporates its local error vector ${{\mathbf{e}}_i^t}$ with the coded vector and compresses it as
\begin{align}
    \label{compress}
    {\mathbf{\hat g}}_i^t = \mathcal{C}\left( {\gamma {\mathbf{g}}_i^t + {\mathbf{e}}_i^t} \right),
\end{align}
where $\mathcal{C}:{\mathbb{R}^D} \to {\mathbb{R}^D}$ is a specific type of biased compression function, and $\gamma $ denotes the learning rate. Here, the error vector represents the compression error accumulated from the previous iterations before obtaining the global model ${{\boldsymbol{\theta}}^t}$. By adding the error back into the coded gradient, the compression bias in the previous iterations can be compensated, which guarantees the correct direction to update the model. Note that the length of the compressed vector ${\mathbf{\hat g}}_i^t$ is the same as the input vector of the compression function, and the communication is compressed based on the fact that less bits are required to transmit the compressed vector than the original vector. Some examples of biased compression functions $\mathcal{C}$ are listed as follows: 
    \begin{itemize}
   \item   \textbf{Grouped sign-bit quantization.}
The input vector is $\mathbf{g} \in \mathbb{R}^D$.
We partition the index set $\{1,\dots,D\}$ into $M_0$ non-overlapping groups.
Group $m$ is denoted by $\mathcal{I}_m$ with cardinality $|\mathcal{I}_m|$, for $m=1,\dots,M_0$.
For each group $m$, we extract the subvector
\[
    \mathbf{g}_m = (g_j)_{j \in \mathcal{I}_m} \in \mathbb{R}^{|\mathcal{I}_m|},
\]
where $g_j$ is the $j$-th element in the input vector $\mathbf{g}$.
We perform the following compression independently on each group:
\begin{align}
\label{signbit_each_group}
    \mathcal{C}_m(\mathbf{g}_m)
    =
    \operatorname{sign}(\mathbf{g}_m)\,
    \frac{\|\mathbf{g}_m\|_1}{|\mathcal{I}_m|},
    \qquad m = 1,\dots,M_0,
\end{align}
 where ${\text{sign}}\left( {\cdot} \right)$ is the sign operator, and ${{{\left\| {\cdot} \right\|}_1}}$ is the $l^1$ norm.
The output of grouped sign-bit quantization is obtained by concatenating the outputs in (\ref{signbit_each_group}):
\begin{align}
    \mathcal{C}(\mathbf{g})
    &=
    \mathcal{C}_1(\mathbf{g}_1) \,\Vert\,
    \mathcal{C}_2(\mathbf{g}_2) \,\Vert\, \cdots \,\Vert\,
    \mathcal{C}_{M_0}(\mathbf{g}_{M_0}).
\end{align}
         Note that, the concatenation preserves the same element order as in the original input vector $\mathbf{g}$. When $M_0=1$, grouped sign-bit quantization is also known as \textbf{sign-bit quantization}. 
        \item \textbf{Top-$K$ sparsification.} In this case, the output of the compression function only remains $K$ elements in the input of the compression function with the largest magnitude while zeroing out the others. 
    \end{itemize}

After the compression, the error vector at non-straggler device $i$ is updated as
\begin{align}
    \label{error update}
    {\mathbf{e}}_i^{t + 1} = \gamma {\mathbf{g}}_i^t + {\mathbf{e}}_i^t - \mathcal{C}\left( {\gamma {\mathbf{g}}_i^t + {\mathbf{e}}_i^t} \right),
\end{align}
which is the difference between the original and the compressed vector. 

Let us use $I_i^t=1$ to indicate that device $i$ is not a straggler in iteration $t$ while $I_i^t=0$ indicates the opposite, where it holds that 
\begin{align}
    \label{straggler model}
    \Pr \left( {I_i^t = 1} \right) = 1 - p,\Pr \left( {I_i^t = 0} \right) =  p,
\end{align}
based on the straggler model introduced in Section~\ref{problem model}. 
Different from the non-stragglers, for the straggler device $i$, where $I_i^t = 0$, no message is successfully transmitted to the server\footnote{Note that the straggling behavior of the stragglers can be caused either by delayed or failed computations or by communication errors.}. In that case, the error vectors on the stragglers remain unchanged, i.e., ${\mathbf{e}}_i^{t + 1} = {\mathbf{e}}_i^t$. After receiving the messages from all non-stragglers, the server aggregates the messages as
\begin{align}
    \label{server agg}
    {{\mathbf{\hat g}}^t} = \sum\limits_{\left\{ {\left. i \right|I_i^t \ne 0} \right\}} {{\mathbf{\hat g}}_i^t},
\end{align}
which is an approximate version of the global gradient. 
Finally, the server updates the global model as
\begin{align}
    \label{update model}
    {{\boldsymbol{\theta}}^{t + 1}} = {{\boldsymbol{\theta }}^t} - {{\mathbf{\hat g}}^t}.
\end{align}
 We describe the proposed method as Algorithm~\ref{Alg 1}, and provide its flowchart in Fig.~\ref{fig: coco-ef}. 
 
% \begin{algorithm}
% \caption{ {COCO-EF}}\label{Alg 1}
% \textbf{Input:} Initial parameter vector \({\boldsymbol{\theta}}^0\), stepsize $\gamma$, and number of iterations $T$\\
% \textbf{1. Before the training:} The training subsets are allocated to the devices in a pairwise balanced scheme \\
% \textbf{2. During training:}\\
% \textbf{Initialize: }$t = 0$, local error vectors ${\mathbf{e}}_i^0 = 0$, $\forall i$\\
%  \While{\(t \leq T \)}{
%  \textcolor{my1}{\textbf{In parallel for all devices} \(i \in \left\{ {1,...,N} \right\}\):\\
%  \eIf{\(I_i^t = 1\)}{
%  \For{\(k \in \left\{ {1,...,M} \right\}\)}{
%         \If{\(s\left( {i,k} \right) = 1\)}{
%   Compute local gradient \( \nabla f_k(\boldsymbol{\theta}_i^{t})\)\;
%    }
%     }
%     Encode the local gradients into ${\mathbf{g}}_i^t$ as (\ref{encoding});\\
%    Incorporate ${\gamma {\mathbf{g}}_i^t + {\mathbf{e}}_i^t}$; \\  
%    Compress to obtain ${\mathbf{\hat g}}_i^t$ as (\ref{compress});\\
%    Update ${\mathbf{e}}_i^{t + 1}$ as (\ref{error update});\\
%    Send ${\mathbf{\hat g}}_i^t$ to the server;\\
%     }
%     {
%    ${\mathbf{e}}_i^{t + 1}={\mathbf{e}}_i^{t}$;
%   }}
%   \textcolor{my2}{\textbf{For the server:}\\
%   Receive the messages and compute ${{\mathbf{\hat g}}^t}$ as (\ref{server agg});\\
%   Update the global model as (\ref{update model});\\}
%   $t=t+1$;\\
% }
% \end{algorithm}
\begin{algorithm}
\label{Alg 1}
 {\caption{ {COCO-EF}}}
 {\textbf{Input:} Initial parameter vector \({\boldsymbol{\theta}}^0\), stepsize $\gamma$, and number of iterations $T$\\
\textbf{1. Before the training:} The training subsets are allocated to the devices in a pairwise balanced scheme \\}
{\textbf{2. During training:}\\
\textbf{Initialize: }$t = 0$, local error vectors ${\mathbf{e}}_i^0 = 0$, $\forall i$\\}
 {\While{\(t \leq T \)}{
 \textbf{For the server:} Broadcast the global model $\boldsymbol{\theta}^t$ to all devices;\\
 {\textbf{In parallel for all devices} \(i \in \left\{ {1,...,N} \right\}\):\\
 \eIf{\(I_i^t = 1\)}{
 \For{\(k \in \left\{ {1,...,M} \right\}\)}{
        \If{\(s\left( {i,k} \right) = 1\)}{
  Compute local gradient \( \nabla f_k(\boldsymbol{\theta}_i^{t})\)\;
   }
    }
    Encode the local gradients into ${\mathbf{g}}_i^t$ as (\ref{encoding});\\
   Incorporate ${\gamma {\mathbf{g}}_i^t + {\mathbf{e}}_i^t}$; \\  
   Compress to obtain ${\mathbf{\hat g}}_i^t$ as (\ref{compress});\\
   Update ${\mathbf{e}}_i^{t + 1}$ as (\ref{error update});\\
   Send ${\mathbf{\hat g}}_i^t$ to the server;\\
    }
    {
   ${\mathbf{e}}_i^{t + 1}={\mathbf{e}}_i^{t}$;
  }}
  {\textbf{For the server:}\\
  Receive the messages and compute ${{\mathbf{\hat g}}^t}$ as (\ref{server agg});\\
  Update the global model as (\ref{update model});\\}
  $t=t+1$;\\
}}
\end{algorithm}
\begin{figure*}
    \centering
    \includegraphics[width=\linewidth]{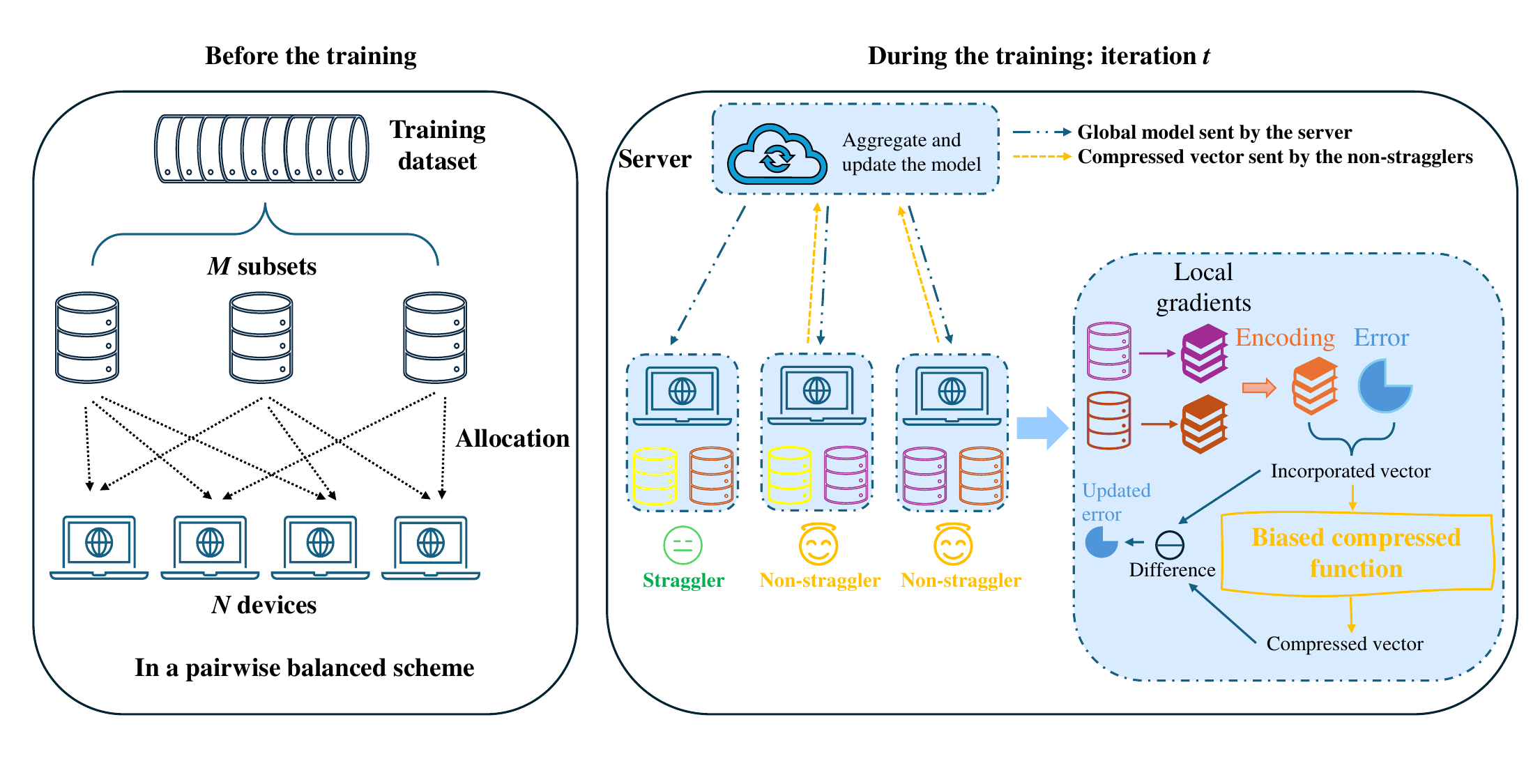}
    \caption{The flowchart of COCO-EF.}
    \label{fig: coco-ef}
\end{figure*}

\section{Convergence Analysis}
\label{performance analysis}
In this section, the convergence performance of COCO-EF is analyzed. Before providing the main theorem, some widely-used assumptions are stated as follows.  
\begin{assumption}
    \label{assp smooth}
The training loss function $F$ is $L$-smooth \cite{karimireddy2019error,ding2023dsgd}, which implies
    \begin{align}
        \label{smooth}
  F\left( {\mathbf{x}} \right) \leq & F\left( {\mathbf{y}} \right) + \left\langle {\nabla F\left( {\mathbf{y}} \right),{\mathbf{x}} - {\mathbf{y}}} \right\rangle  + \frac{L}{2}{\left\| {{\mathbf{x}} - {\mathbf{y}}} \right\|^2},\forall {\mathbf{x}},{\mathbf{y}},\\
      \label{smooth L}
  &    \left\| {\nabla F\left( {\mathbf{x}} \right) - \nabla F\left( {\mathbf{y}} \right)} \right\| \leqslant L\left\| {{\mathbf{x}} - {\mathbf{y}}} \right\|,\forall {\mathbf{x}},{\mathbf{y}}.
  \end{align}
 Here, ${\left\| {\cdot} \right\|}$ denotes the squared $\ell^2$-norm, i.e., $\|\mathbf{x}\|^2 = \sum_{n=1}^{D} x_n^2$, where $x_n$ is the $n$-th element in $\mathbf{x}\in \mathbb{R}^D$.
\end{assumption}
\begin{assumption}
    \label{assp bounded heter}
    The heterogeneity among $\left\{ {{\mathcal{W}_1},...,{\mathcal{W}_M}} \right\}$ is bounded \cite{karimireddy2022byzantinerobust,ding2023dsgd,zhu2023byzantine,islamov2024asgrad}, which indicates  
    \begin{align}
        \label{bounded heterogenity}
        \left\| {\nabla {f_k}\left( {\boldsymbol{\theta}} \right) - \frac{1}{M}\nabla F\left( {\boldsymbol{\theta}} \right)} \right\|^2 \leqslant {\beta ^2},\forall k,\forall {\boldsymbol{\theta}}.
    \end{align}
\end{assumption}
\begin{assumption}
    \label{assp lower bound}
    The training loss is lower bounded by some constant ${{F^*}}$, which means \cite{koloskova2019decentralized,jin2024sign}
    \begin{align}
        \label{lower bound of F}
        F\left( {\boldsymbol{\theta}} \right) \geqslant {F^*},\forall \boldsymbol{\theta}.
    \end{align} 
\end{assumption}
\begin{assumption}
    \label{assp agg error}
    The compression discrepancy is bounded by \cite{li2023analysis} 
    \begin{align}
        \label{bounded agg}
        \mathbb{E}\left[ {{{\left\| {\sum\limits_{i = 1}^N {\left[ {{{\mathbf{x}}_i} - \mathcal{C}\left( {{{\mathbf{x}}_i}} \right)} \right]} } \right\|}^2}} \right] \leqslant {q_A} {{{\left\| {\sum\limits_{i = 1}^N {{{\mathbf{x}}_i}} } \right\|}^2}} ,\forall {{\mathbf{x}}_i},
    \end{align}
where $q_A<1$ is a constant. 
\end{assumption}
\begin{assumption}
    \label{assp compre}
    The biased compression function $\mathcal{C}$ satisfies \cite{karimireddy2019error} 
    \begin{align}
        \label{bounded compression}
        \mathbb{E}\left[ {{{\left\| {\mathcal{C}\left( {\mathbf{x}} \right) - {\mathbf{x}}} \right\|}^2}} \right] \leqslant \delta {\left\| {\mathbf{x}} \right\|^2},\forall {\mathbf{x}},
    \end{align}
    where $0 \leqslant \delta  < 1$ is a constant.
\end{assumption}

         Note that, in Assumption~\ref{assp agg error} and Assumption~\ref{assp compre}, the expectation is taken over the randomness of the compression function $\mathcal{C}$, which may be possibly random. Although the grouped sign-bit quantization and top-$K$ sparsification introduced in Section~\ref{our method} are both deterministic compression functions, we retain the expectation operator to make these assumptions applicable to a broader variety of biased compression functions.
         Based on Assumption~\ref{assp agg error} and Assumption~\ref{assp compre}, we can derive the following propositions.

    \begin{proposition}
    The parameter $q_A$ in Assumption~\ref{assp agg error} depends on the value of $\delta$, where a larger value of $\delta$ indicates a higher level of information loss caused by the compression. To illustrate this, consider the special case where $\delta$ is very close to zero. In this case, the information loss due to compression is negligible, and $q_A$ would be close to zero as well. 
    \end{proposition}
     \begin{proposition}
     \label{prop2}
     For grouped sign-bit quantization introduced in Section~\ref{our method}, the constant $\delta$ defined in Assumption~\ref{assp compre} can be expressed as \(\delta  = 1 - {\min _{m \in \left\{ {1,...,{M_0}} \right\}}}\frac{1}{{\left| {{\mathcal{I}_m}} \right|}}\) \cite{li2023analysis}. For top-$K$ sparsification, $\delta$ is given as $\delta=1-\frac{K}{D}$ \cite{li2023analysis}, where $D$ is the length of the input vector of the compression function. In \cite{li2023analysis}, it has been demonstrated that the empirical value of $q_A$ remains well bounded by 1 during training, whether sign-bit quantization or top-$K$ sparsification is adopted as the compression function in the DL setting.
     \end{proposition}

Next, based on Assumptions \ref{assp smooth}-\ref{assp compre}, we present two lemmas to aid the derivation of the main theorem. 
\begin{lemma}
    \label{lemma 1}
    The following bound can be provided:
    \begin{align}
        \label{sum ig}
  &\mathbb{E}\left[ {\left. {{{\left\| {\sum\limits_{i = 1}^N {I_i^t{\mathbf{g}}_i^t} } \right\|}^2}} \right|{\mathcal{F}^t}} \right] \nonumber \\
   \leqslant& \frac{{2p{\beta ^2}\vartheta }}{{1 - p}}  + \left[ {\frac{p}{{\left( {1 - p} \right)N}} + 1 + \frac{{2p\vartheta }}{{\left( {1 - p} \right){M^2}}}} \right]{\left\| {\nabla F\left( {{{\boldsymbol{\theta }}^t}} \right)} \right\|^2},
    \end{align}
where $\mathbb{E}\left[ {\left. \cdot \right|{\mathcal{F}^t}} \right]$ is the expectation conditioned on iterations $0,...,t-1$, and 
\begin{align}
    \label{def vartheta}
    \vartheta  \triangleq \sum\limits_{k = 1}^M {\left( {\frac{1}{{{d_k}}} - \frac{1}{N}} \right)}.
\end{align}
\end{lemma}
\begin{proof}
    Please see Appendix~\ref{appendix lemma 1}. 
\end{proof}
\begin{lemma}
    \label{lemma error}
    For $\delta  < 0.5$ and ${q_A} < \frac{{2\delta  + 1}}{2}$, the error vectors can be bounded as 
    \begin{align}
        \label{bound error}
       & \sum\limits_{t = 0}^T {\mathbb{E}\left[ {{{\left\| {\sum\limits_{i = 1}^N {{\mathbf{e}}_i^{t + 1}} } \right\|}^2}} \right]} \nonumber\\
        \leqslant & \left( {T + 1} \right){\gamma ^2}{\xi _1} + {\gamma ^2}{\xi _2}\sum\limits_{t  = 0}^T {\mathbb{E}\left[ {{{\left\| {\nabla F\left( {{{\boldsymbol{\theta }}^t }} \right)} \right\|}^2}} \right]},
    \end{align}
    where $\xi_1$ and $\xi_2$ are defined as (\ref{xi1}) and (\ref{xi2}).
    \begin{figure*}
        \begin{align}
        \label{xi1}
       {\xi _1} \triangleq \frac{1}{{1 - \left[ {\frac{{\left( {1 - p} \right)\left( {2\delta  + 1} \right)}}{2} + p} \right]}}\left\{ {\frac{{8{\beta ^2}p\delta \vartheta }}{{1 - p}} + \frac{{\left( {4\delta  + 2} \right)4\delta p{\beta ^2}\vartheta }}{{2\left( {1 - p} \right)\delta  + p}}\frac{1}{{1 - \left[ {2\left( {1 - p} \right)\delta  + p} \right]}}} \right\}.
    \end{align}
    \begin{align}
        \label{xi2}
       & {\xi _2} \triangleq \nonumber\\
        &\left\{ {\frac{{4p\delta }}{{1 - p}}\left( {\frac{1}{N} + \frac{{2\vartheta }}{{{M^2}}}} \right) + \frac{{{q_A}\left( {2\delta  + 1} \right)}}{{\left( {1 - p} \right)\left( {2\delta  + 1 - 2{q_A}} \right)}} + \frac{{\left[ {\frac{{\left( {1 - p} \right)\left( {2\delta  + 1} \right)}}{2} + p} \right]\frac{{\left( {4\delta  + 2} \right)2\delta p}}{{2\left( {1 - p} \right)\delta  + p}}\left( {\frac{1}{N} + \frac{{2\vartheta }}{{{M^2}}}} \right)}}{{\frac{{\left( {1 - p} \right)\left( {2\delta  + 1} \right)}}{2} + p - \left[ {2\left( {1 - p} \right)\delta  + p} \right]}}} \right\}\frac{1}{{1 - \left[ {\left( {1 - p} \right)\frac{{2\delta  + 1}}{2} + p} \right]}}.
    \end{align}
    \end{figure*}
    
\end{lemma}
\begin{proof}
    Please see Appendix~\ref{appendix lemma error}. 
\end{proof}
Based on Lemma~\ref{lemma 1} and Lemma~\ref{lemma error}, the convergence performance of COCO-EF is characterized in Theorem 1. 
\begin{theorem}[Convergence performance of COCO-EF]
\label{convergence performance}
Based on Assumptions~\ref{assp smooth}-\ref{assp compre}, for $\delta  < 0.5$ and ${q_A} < \frac{{2\delta  + 1}}{2}$, with  a constant  learning rate $\gamma  = \frac{\phi }{{\sqrt {T + 1} }}$, $\phi>0$, for \(T > {\left( {{\varepsilon _0}\phi } \right)^2} - 1\), COCO-EF converges as
\begin{align}
    \label{conv}
    &\frac{1}{{T + 1}}\sum\limits_{t = 0}^T {\mathbb{E}\left[ {{{\left\| {\nabla F\left( {{{\boldsymbol{\theta }}^t}} \right)} \right\|}^2}} \right]} \nonumber\\
    \leqslant & \frac{{{\varepsilon _1}\phi }}{{\sqrt {T + 1}  - {\varepsilon _0}\phi }} + \frac{{ {F\left( {{{\boldsymbol{\theta}}^0}} \right)}  - {F^*}}}{{\phi \sqrt {T + 1}  - {\varepsilon _0}{\phi ^2}}},
\end{align}
where
\begin{align}
    \label{def varepsilon0}
  {\varepsilon _0} & \triangleq \frac{L}{2}\left[ {\frac{p}{{\left( {1 - p} \right)N}} + 1 + \frac{{2p\vartheta }}{{\left( {1 - p} \right){M^2}}}} \right] + \frac{{L{\xi _2}\beta \sqrt {p\vartheta } }}{{\sqrt {2{\xi _1}\left( {1 - p} \right)} }} \nonumber \\
   &+ \frac{{L\sqrt {{\xi _1}\left( {1 - p} \right)} }}{{2\sqrt {2p{\beta ^2}\vartheta } }}\left[ {\frac{p}{{\left( {1 - p} \right)N}} + 1 + \frac{{2p\vartheta }}{{\left( {1 - p} \right){M^2}}}} \right],\\
    \label{def varepsilon1}
    {\varepsilon _1} &\triangleq \sqrt {\frac{{2{L^2}{\xi _1}p{\beta ^2}\vartheta }}{{1 - p}}}  + \frac{{Lp{\beta ^2}\vartheta }}{{1 - p}},
\end{align}
 and $T$ is the total number of iterations. 
\end{theorem}
\begin{proof}
    Please see Appendix~\ref{proof th}. 
\end{proof}
 Note that the condition in Theorem~\ref{convergence performance}, i.e., $\delta < 0.5$ and $q_A < \frac{2\delta + 1}{2}$, indicates that the information loss caused by compression is upper-bounded. These conditions can be readily satisfied by adjusting the compression level of the compression functions. This can be achieved by tuning the parameters of the compression functions, such as $K$ in top-$K$ sparsification and $M_0$ together with the group size in grouped sign-bit quantization.

In Theorem~\ref{convergence performance}, (\ref{conv}) can be rewritten as 
    \begin{align}
        \label{eq O}
 & \frac{1}{{T + 1}}\sum\limits_{t = 0}^T {\mathbb{E}\left[ {{{\left\| {\nabla F\left( {{{\boldsymbol{\theta }}^t}} \right)} \right\|}^2}} \right]}  \nonumber \\
   = & O\left[ {\frac{{{\varepsilon _1}\phi }}{{\sqrt {T + 1} }} + \frac{{F\left( {{{\boldsymbol{\theta }}^0}} \right) - {F^*}}}{{\phi \sqrt {T + 1} }}} \right],
    \end{align}
where a smaller value of $\varepsilon _1$ indicates a better learning performance. Based on the expression of $\varepsilon _1$ provided in (\ref{def varepsilon1}), we can make the following observations:
\begin{itemize}
    \item For a smaller value of the straggler probability \( p \), indicating that devices are less likely to become stragglers, the value of \( \varepsilon_1 \) decreases, improving the learning performance. This aligns with our intuition that better learning performance can be achieved with less missing information from stragglers.
    \item If we increase the values of \(\left\{ {d_k, k = 1, \ldots, M} \right\}\), implying that the training subsets are allocated to devices more redundantly, the value of \( \vartheta \) in (\ref{def vartheta}) decreases, leading to a reduction in \( \varepsilon_1 \). In this case, learning performance improves at the cost of increased computational and storage burdens on the devices. 
    \item With a smaller value of \( L \), indicating that the training loss function is smoother, the value of \( \varepsilon_1 \) decreases.  Noting that the value of $L$ affects only the term $\varepsilon_1$ on the right-hand side of (\ref{eq O}), a smaller value of $L$ implies improved learning performance.  The rationale is as follows. In COCO-EF, information on the correct direction of the model update, contained in the error vectors, is utilized in a delayed manner in later iterations. For a smoother loss function, the gradients change more gradually, making the delayed information in the error vectors less outdated, which improves learning performance more effectively. 
    \item As the value of \( \beta \) decreases, the value of \( \varepsilon_1 \) also decreases. In other words, if the training data are divided into more homogeneous subsets, the learning performance can be improved. The rationale is that when subsets are more similar, missing information from stragglers can be more effectively compensated by information from non-stragglers, given that the local gradients computed from different subsets are more similar.
    \item When \( \delta = 0 \), corresponding to the case without communication compression, \( \varepsilon_1 = 0 \), and the first term on the right-hand side of (\ref{conv}) equals zero. In this case, the learning performance represents the optimal performance bound for COCO-EF.
\end{itemize}

\section{Numerical Results}
\label{simulations}
In this section, the performance of COCO-EF is evaluated on different learning tasks.
\subsection{Linear regression task}
We consider a linear regression task with a synthetic dataset. Suppose there are \(N=100\) devices and 
the overall dataset \(\mathcal{W}\) is divided into \(M=100\) subsets, where each subset 
\(\mathcal{W}_k\) contains a single data sample \(\left\{ {{{\mathbf{z}}_k},{y_k}} \right\}\). 
Here, \(y_k \in {\mathbb{R}}\) represents the label of \({{\mathbf{z}}_k} \in {\mathbb{R}^{100}}\). The elements in \( \left\{ \mathbf{z}_1, \dots, \mathbf{z}_M \right\} \) are sampled independently from the normal distribution \( \mathcal{N}(0, 100) \). To generate each \( y_k \), we first create a random vector \( \hat{\boldsymbol{\theta}} \) consisting of 100 elements drawn from a standard normal distribution. Consequently, each \( y_k \) is generated according to the distribution \( y_k \sim \mathcal{N}(\langle \mathbf{z}_k, \hat{\boldsymbol{\theta}} \rangle, 1) \),  \( \forall k \).
Based on that, the loss function is given as 
\begin{align}
\label{LR loss}
F\left( {\boldsymbol{\theta}} \right)&= \sum\limits_{k = 1}^M {{f_k}\left( {\boldsymbol{\theta}} \right)},\nonumber\\
{f_k}\left( {\boldsymbol{\theta}} \right)& = \frac{1}{2}{\left( {\left\langle {{\boldsymbol{\theta}},{{\mathbf{z}}_k}} \right\rangle  - {y_k}} \right)^2},
\end{align}
where \(k=1,...,M\), \(\boldsymbol{\theta} \in {\mathbb{R}^{100}}\). 

Before the training starts, we allocate the training subsets uniformly and randomly to all devices, in which case each subset $\mathcal{W}_k$ is allocated to \(d_k\) devices, $\forall k$. This can be regarded as an approximation of training data allocation in a pairwise balanced scheme, as pointed out in \cite{bitar2020stochastic}. The parameter vector is initialized as \( {\boldsymbol{\theta}}^0 \), where each element is drawn from the standard normal distribution independently. 

First, to verify the superiority of the proposed method, we compare the performance of the following methods: 
\begin{enumerate}
    \item \textbf{COCO-EF (Sign):} The proposed method that uses sign-bit quantization as the compression function. In this setting, sign-bit quantization reduces to the grouped sign-bit quantization with $M_0=1$, as described in Section~\ref{our method}.
    \item \textbf{COCO-EF (Top-$K$):} The proposed method that uses top-$K$ sparsification described in Section~\ref{our method} as the compression function. 
    \item \textbf{Unbiased (Sign):} The 1-bit gradient coding method proposed in \cite{chengxi20241-BitGC}, which combines the advantages of gradient coding and 1-bit stochastic quantization to deal with the stragglers with communication bottleneck in DL. 
   \item \textbf{Unbiased (Rand-$K$):} This method extends the method proposed in \cite{chengxi20241-BitGC} by using amplified rand-$K$ sparsification as the compression function in place of 1-bit stochastic quantization.
   \item   \textbf{Unbiased-diff (Sign):} This method combines Unbiased (Sign) with gradient-difference compression \cite{stich2020communication} to compensate for the compression error. 
   \item   \textbf{Unbiased-diff (Rand-$K$):} This method combines Unbiased (rand-$K$) with gradient-difference compression \cite{stich2020communication} to compensate for the compression error. 
\end{enumerate}
\begin{figure}[tb]
    \centering
        \includegraphics[width=0.8\linewidth]{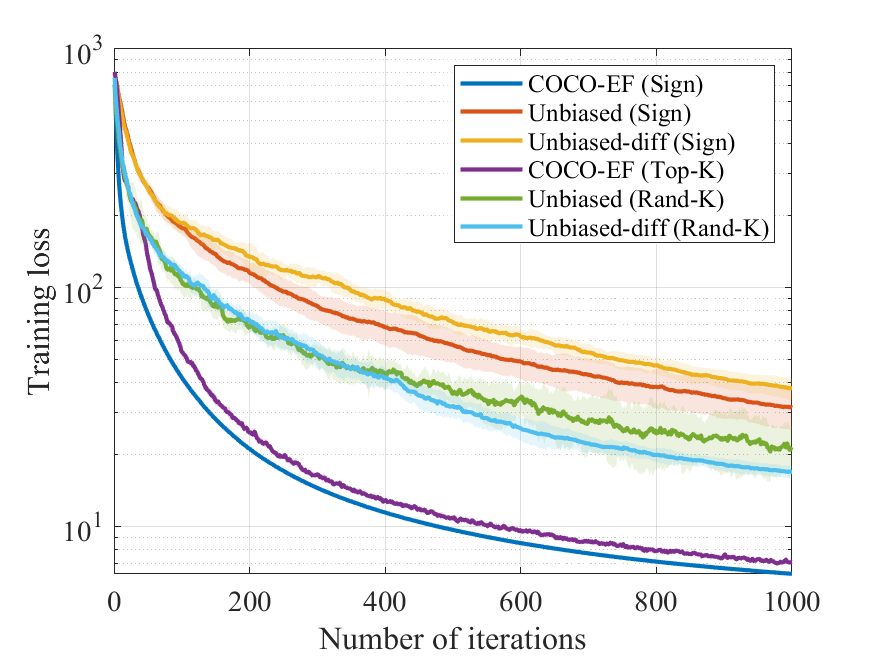}
    \caption{Training loss as a function of the number of iterations for COCO-EF and the baselines with various compression functions. For each method, we run 5 independent trials. The solid curve shows the mean training loss as a function of the number of iterations, and the shaded region represents the standard deviation across trials.}
    \label{fig: biased_unbiased}
    %\vskip 0.2in
\end{figure}
In the methods above, COCO-EF (Sign) and COCO-EF (Top-\(K\)) are different realizations of the proposed method, using two types of biased compression functions. Unbiased (Sign) and Unbiased (Rand-\(K\)) are baseline methods designed to address the same problem considered in this paper, but they use unbiased compression functions for communication compression \cite{chengxi20241-BitGC}.
 Unbiased-diff (Sign) and Unbiased-diff (Rand-\(K\)) are two baseline methods that combine Unbiased (Sign) and Unbiased (Rand-\(K\)) with gradient-difference compression in \cite{stich2020communication}, respectively, in order to mitigate the negative impact of the compression error on the learning performance. 
  Note that in COCO-EF (Sign), Unbiased (Sign) and Unbiased-diff (Sign), 1-bit vectors are transmitted from the non-stragglers to the server, while in COCO-EF (Top-\(K\)), Unbiased (Rand-\(K\)) and Unbiased-diff (Rand-\(K\)), sparsified vectors are sent by the non-stragglers. Based on this, it is rational to compare the learning performance among COCO-EF (Sign), Unbiased (Sign) and Unbiased-diff (Sign), as well as among COCO-EF (Top-\(K\)), Unbiased (Rand-\(K\)) and Unbiased-diff (Rand-\(K\)), to demonstrate the performance gain of the proposed method resulting from the use of biased compression functions in gradient coding.  
 Here, we only compare the proposed method with different realizations of the baseline method in \cite{chengxi20241-BitGC} and omit the comparison with the method in \cite{bitar2020stochastic} mentioned in Section~\ref{introduction}. This is because the superiority of the method in \cite{chengxi20241-BitGC} over \cite{bitar2020stochastic} has already been thoroughly demonstrated in \cite{chengxi20241-BitGC}. Based on this, it is sufficient for the present work to show the superiority of our method over \cite{chengxi20241-BitGC} in order to demonstrate its value over all these prior works. 
In Fig.~\ref{fig: biased_unbiased}, we plot the training loss as a function of the number of iterations, where we set $d_k=5, \forall k$, $p=0.2$ and $K=2$. The learning rate of COCO-EF is fixed at $\gamma=10^{-5}$ and the learning rates of the baselines are fine-tuned to be $\gamma=2\times10^{-6},10^{-5},2\times10^{-6},6\times10^{-6}$ for Unbiased (Sign), Unbiased (Rand-\(K\)), Unbiased-diff (Sign), and Unbiased-diff (Rand-\(K\)). 
 Here, when selecting the learning rate for each method, we fine-tune this parameter to ensure that the best possible performance is achieved for that method for a fair comparison.  
 From Fig.~\ref{fig: biased_unbiased}, we observe that the proposed method achieves a lower training loss after the same number of iterations. Since the communication overhead per iteration is identical among COCO-EF (Sign), Unbiased (Sign) and Unbiased-diff (Sign), as well as among COCO-EF (Top-$K$), Unbiased (Rand-$K$) and Unbiased-diff (Rand-$K$), it is clear that the proposed method outperforms the baselines, regardless of whether 1-bit messages or sparsified messages are used. In other words, the proposed method delivers better learning performance under the same communication overhead in the presence of stragglers. This demonstrates that biased compression functions are superior to unbiased compression functions for communication compression in gradient coding for DL, under communication bottlenecks with stragglers. 
 Also note that, under the parameter settings used here and according to Proposition~\ref{prop2}, the condition $\delta < 0.5$ is not satisfied in our experiments. However, the proposed method still converges empirically under the experimental configurations we consider. This suggests that the theoretical requirement in Theorem~1 is a sufficient but not a necessary condition for convergence. 

 It is worth noting that when unbiased compression is adopted in the baseline methods, some existing techniques can potentially be used to compensate for the compression error, such as error feedback and gradient-difference compression. However, it has been demonstrated in \cite{horvath2020better} that error feedback provides little improvement in learning performance when used with unbiased compression in DL settings. To further verify this behavior in our context, we conducted additional simulations by combining Unbiased (Sign) and Unbiased (Rand-\(K\)) with error feedback for the considered problem. The results show that the method barely converges under this configuration. Therefore, we do not include these simulation results in the paper.
Based on these observations, it becomes clear that the performance gain of the proposed method arises from the combination of biased compression and error feedback, rather than from the advantage of error feedback alone. In other words, the benefits of error feedback are effective only when used with biased compression, and can therefore be regarded as part of the overall advantage of biased compression in the proposed method.
In addition, our experiments show that, in the baseline, the performance gain brought by gradient-difference compression is negligible, and the resulting learning performance is inferior to that of the proposed method. This is due to the inherent drawback of unbiased compression in DL settings—its average ability to retain information is lower than that of biased compression \cite{beznosikov2023biased}.

\begin{figure}[tb]
    \centering
        \includegraphics[width=0.7\linewidth]{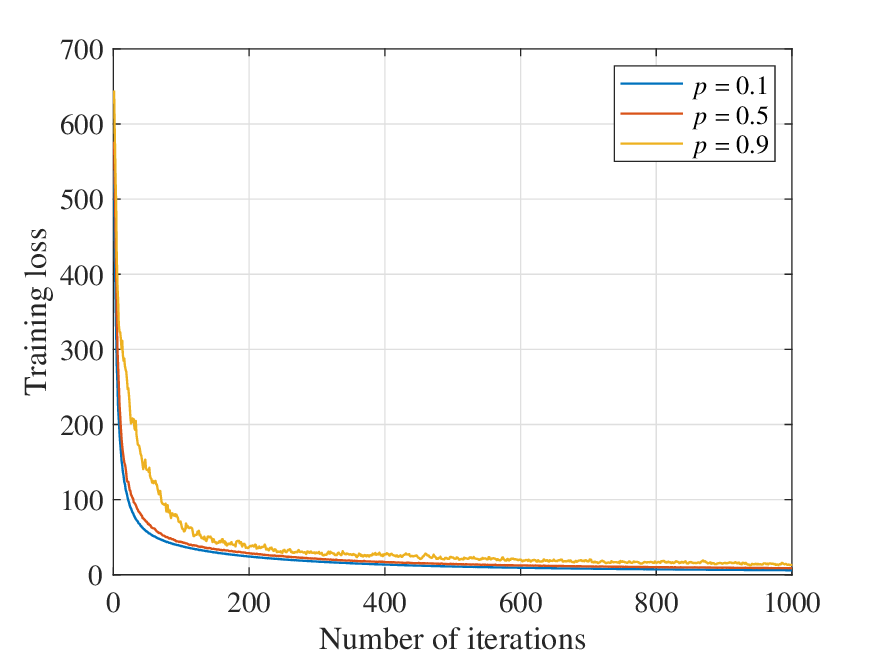}
    \caption{Training loss as a function of the number of iterations for COCO-EF (Sign) under varying values of $p$.}
    \label{fig: varying p}
\end{figure}

Next, we investigate the influence of the straggler probability $p$ on the learning performance of the proposed method. To this end, in Fig.~\ref{fig: varying p}, we plot the training loss as a function of the number of iterations under different values of $p$, where COCO-EF (Sign) is implemented. Here we fix $d_k=2, \forall k$, and $\gamma=10^{-5}$. From Fig.~\ref{fig: varying p}, we observe that with a larger value of \( p \), indicating that devices are more likely to be stragglers, the learning performance degrades to some extent. This aligns with the theoretical analysis in Section~\ref{performance analysis} and our intuition that more stragglers lead to more missing information, which cannot be utilized by the server to update the global model effectively and accurately. However, performance degradation only becomes noticeable when \( p \) is fairly large (i.e., close to 1), which can be attributed to the gradient coding scheme adopted in COCO-EF based on a redundant allocation of training data across devices, enhancing robustness against stragglers. 

\begin{figure}[tb]
    \centering
    \includegraphics[width=0.7\linewidth]{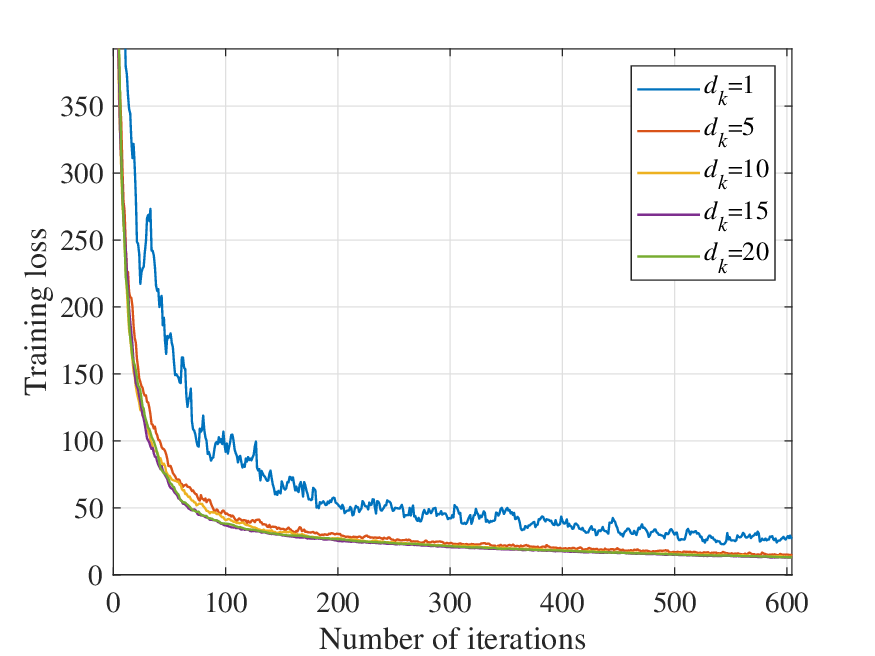}
    \caption{Training loss as a function of the number of iterations for COCO-EF (Sign) under varying values of $d_k$.}
    \label{fig: varying dk}
\end{figure}

\begin{figure}
    \centering
        \includegraphics[width=0.7\linewidth]{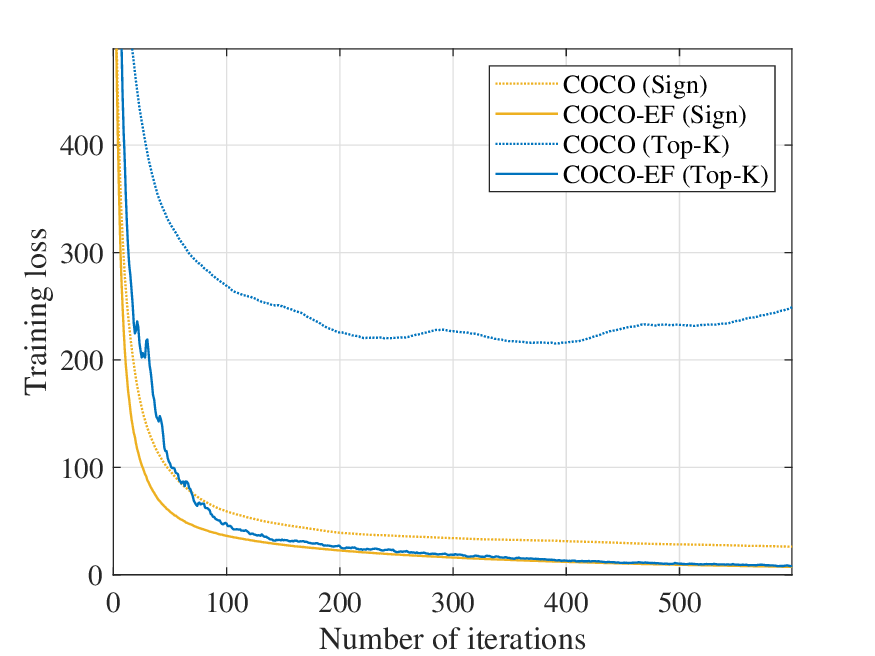}
    \caption{Training loss as a function of the number of iterations for COCO-EF and COCO.}
    \label{fig: ef}
\end{figure}

To illustrate the influence of data allocation redundancy on robustness to stragglers in COCO-EF, Fig.~\ref{fig: varying dk} shows the training loss as a function of the number of iterations for COCO-EF (Sign) under varying values of \( d_k \), with \( p = 0.9 \) and \( \gamma = 10^{-5} \). From Fig.~\ref{fig: varying dk}, it is evident that as \( d_k \) increases, indicating greater redundancy in data allocation, the learning performance of COCO-EF improves, which aligns with our analysis in Section~\ref{performance analysis}. The rationale is that with increased redundancy, the missing information from stragglers can be compensated to a larger extent by information from non-stragglers, enhancing learning performance. It is important to note that improving learning performance by increasing redundancy comes at the cost of greater computational and communication burdens on devices, and this trade-off should be considered in practice. Additionally, from Fig.~\ref{fig: varying dk}, we observe that increasing \( d_k \) from 1 to 10 significantly improves learning performance. However, further increases in \( d_k \) beyond 10 result in only marginal improvements. Note that by setting \( d_k=M \), the effects of stragglers can be completely mitigated, since each device holds a copy of the overall training dataset. This suggests that the proposed method can effectively mitigate the impact of stragglers with a reasonable level of data allocation redundancy. Even when the goal is to fully counteract the effects of stragglers, only a relatively low level of redundancy is necessary.

To verify the necessity of adopting the error feedback mechanism in the proposed method, Fig.~\ref{fig: ef} shows the training loss as a function of the number of iterations for COCO-EF (Sign), COCO (Sign), COCO-EF (Top-\(K\)), and COCO (Top-\(K\)), where we set \( K = 2 \), \( p = 0.2 \), \( d_k = 5 \), and \( \gamma = 10^{-5} \). Here, COCO represents a version of the proposed method without the error feedback mechanism, implemented by fixing \(\left\{ {\mathbf{e}}_i^t = \mathbf{0}, \forall i, \forall t \right\}\). COCO-EF (Sign) and COCO (Sign) use the same compression function, as do COCO-EF (Top-\(K\)) and COCO (Top-\(K\)).
From Fig.~\ref{fig: ef}, we observe that the learning performance of COCO-EF (Sign) is significantly better than that of COCO (Sign). Similarly, COCO-EF (Top-\(K\)) outperforms COCO (Top-\(K\)), with COCO (Top-\(K\)) particularly struggling to converge, while the proposed method consistently converges to the stationary point. This suggests that adopting the error feedback mechanism in the proposed method is essential to ensure convergence and accelerate it by compensating for the compression bias, thus maintaining the correct direction to update the global model.

In our theoretical analysis and in the numerical settings above, we consider the use of a constant learning rate. However, we are also interested in comparing the learning performance of the proposed method under constant and decaying learning-rate schemes. To assess the influence of these different learning-rate settings, Fig.~\ref{fig: decaying lr} plots the training loss as a function of the number of iterations, where COCO-EF (Sign) is implemented. Here, we fix $p = 0.5$ and $d_k = 2$, $\forall k$. Under the constant learning-rate setting, the learning rate is fixed at $\gamma = 2 \times 10^{-5}$. Under the decaying learning-rate setting, the learning rate at iteration $t$ is set to $\gamma^t = \frac{2 \times 10^{-5}}{\sqrt{t+1}}$. This ensures a fair comparison between the constant and decaying schemes, as both begin with the same initial learning rate. 

From Fig.~\ref{fig: decaying lr}, it can be seen that the learning performance is significantly better under the constant learning-rate setting. The rationale behind this phenomenon is as follows. From (\ref{compress}), the encoded vector is scaled by the learning rate at the current iteration $t$. In contrast, the error term in (\ref{compress}) is accumulated from previous iterations, which was scaled by the learning rates used in those earlier iterations. Since the learning rates in the previous iterations are larger than the current learning rate under the decaying learning-rate scheme, the error term in (\ref{compress}) becomes more dominant than the newly encoded vector. This dominance may negatively affect the learning process, as the system tends to overemphasize compensating for past errors rather than striking an appropriate balance between correcting past errors and incorporating the current encoded vectors derived from the latest gradients.

\begin{figure}
    \centering
        \includegraphics[width=0.7\linewidth]{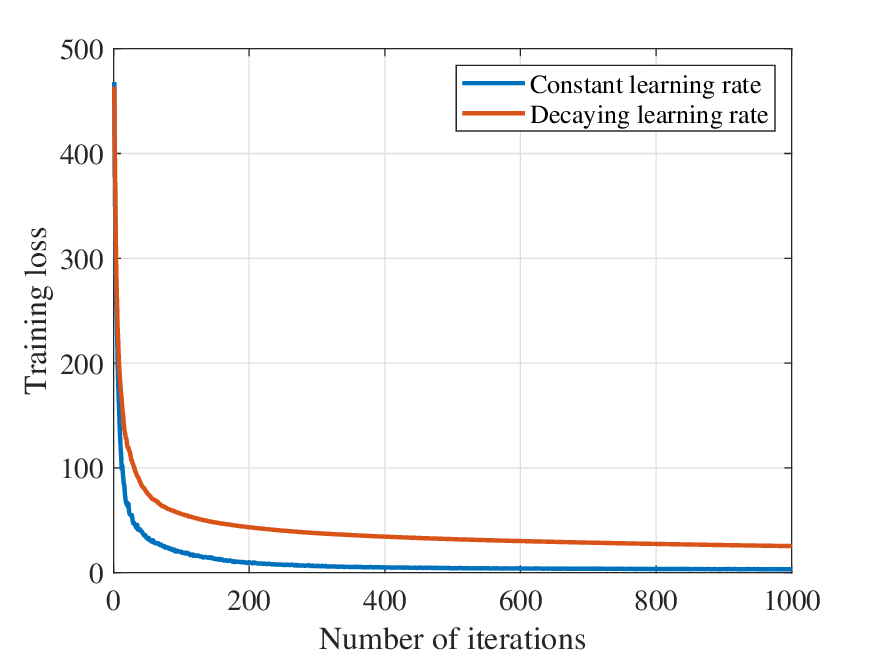}
    \caption{{Training loss as a function of the number of iterations for COCO-EF (Sign) under constant and decaying learning-rate schemes.}}
    \label{fig: decaying lr}
\end{figure}
 
\subsection{Image classification task}
In this subsection, we evaluate the performance of COCO-EF on an image classification task using the MNIST dataset \cite{lecun1998mnist}, where a convolutional neural network is trained. The training loss is defined as the cross-entropy loss, which is a non-convex function. We consider a DL system with \(N = 100\) devices. The MNIST training set contains 60{,}000 samples representing 10 handwritten digits, which are randomly divided into \(M = 100\) subsets without overlap such that each subset contains the same number of samples. To simulate heterogeneity among the subsets, all samples within each subset represent the same digit, while the digits across different subsets may differ. Before training begins, the subsets are allocated to the devices uniformly and randomly; as a result, each subset is assigned to \(d_k\) devices for all \(k\). The straggler probability is fixed at \(p = 0.6\). The MNIST test set is used to evaluate the test loss and test accuracy of the trained model. It contains 10{,}000 samples representing 10 handwritten digits.

To demonstrate the superiority of the proposed method, we compare the performance of COCO-EF (Sign) with that of Unbiased (Sign). In Fig.~\ref{fig:mnist}, we plot the training loss, training accuracy, test loss, and test accuracy as functions of the number of iterations for both methods under varying values of \(d_k\). The learning rates of COCO-EF (Sign) and Unbiased (Sign) are all fine-tuned to be nearly optimal. Note that, for both methods, the communication overhead per iteration is the same. Therefore, comparing their learning performance under the same number of iterations provides a fair comparison under equal communication overhead in the presence of stragglers. From Fig.~\ref{fig:mnist}, it can be seen that the proposed method outperforms the baseline method in all considered settings, demonstrating superior communication efficiency and robustness to stragglers. It can also be observed that the learning performance improves as the value of \(d_k\) increases. This observation aligns with our intuition that increasing data allocation redundancy enhances robustness to stragglers by compensating for the information loss caused by straggling devices through the messages received from non-straggling devices.
 
\begin{figure*}
    \centering
        \includegraphics[width=0.7\linewidth]{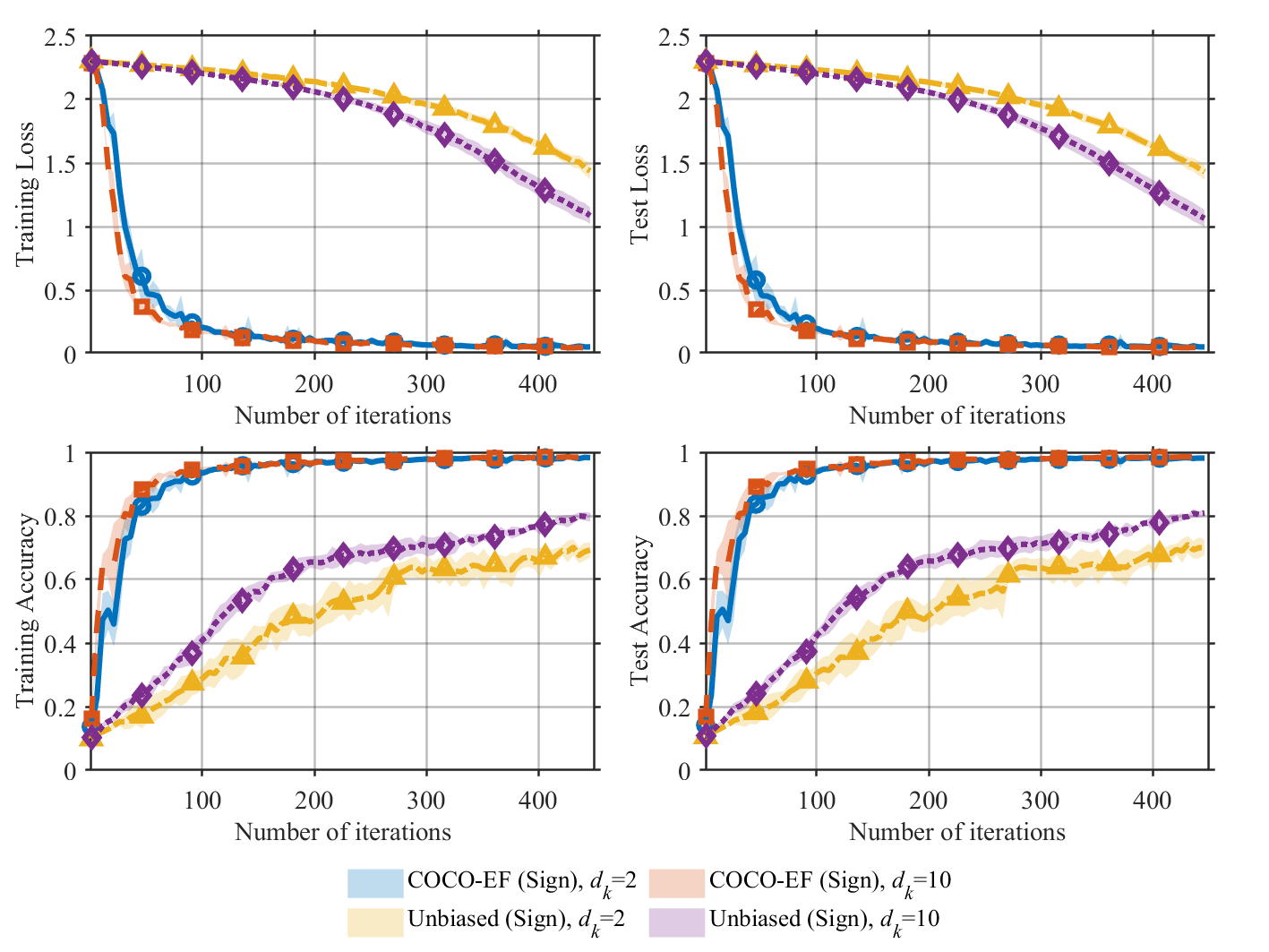}
    \caption{{Training loss, training accuracy, test loss, and test accuracy as functions of the number of iterations for COCO-EF (Sign) and Unbiased (Sign) under different values of \(d_k\). For each method, we conduct 5 independent trials. The solid curves show the mean performance across trials, and the shaded regions represent the corresponding standard deviations.}}
    \label{fig:mnist}
\end{figure*}
\section{Conclusions}
\label{conclusions} 
In this paper, we studied the DL problem with stragglers under the communication bottleneck. We proposed COCO-EF, which combines the benefits of biased compression functions with error feedback and gradient coding. In COCO-EF, each non-straggler device, holding redundantly allocated training data, encodes its local gradients into a single coded vector. This coded vector is incorporated with the compression error before being further compressed and transmitted to the server with a specified biased compression function. With the received vectors and the redundancy in training data allocation, the server can approximately reconstruct the global gradient to update the model. We analyzed the convergence performance of COCO-EF theoretically and presented numerical results, demonstrating its superior learning performance compared to the baselines under the same communication overhead with stragglers.  In our experiments, the number of devices is set to 100, which is practical in many real-world DL settings, such as edge computing systems with tens to hundreds of edge nodes~\cite{chen2025decentralized,tang2025verifiable}. In the future, we plan to apply the proposed method to more complex applications with a substantially larger number of devices in very large-scale learning systems.  In addition, we also plan to explore more advanced error feedback mechanisms to further enhance the performance of the proposed method, such as the one designed in \cite{richtarik2021ef21}.

\bibliographystyle{IEEEtran}   
\bibliography{reference}

%%%%%%%%%%%%%%%%%%%%%%%%%%%%%%%%%%%%%%%%%%%%%%%%%%%%%%%%%%%%%%%%%%%%%%%%%%%%%%%
%%%%%%%%%%%%%%%%%%%%%%%%%%%%%%%%%%%%%%%%%%%%%%%%%%%%%%%%%%%%%%%%%%%%%%%%%%%%%%%
% APPENDIX
%%%%%%%%%%%%%%%%%%%%%%%%%%%%%%%%%%%%%%%%%%%%%%%%%%%%%%%%%%%%%%%%%%%%%%%%%%%%%%%
%%%%%%%%%%%%%%%%%%%%%%%%%%%%%%%%%%%%%%%%%%%%%%%%%%%%%%%%%%%%%%%%%%%%%%%%%%%%%%%
\section{Acknowledgment}
The authors gratefully acknowledge Martin Jaggi who suggested to look at error feedback in the setting of this work.

\appendices
\section{Proof of Lemma \ref{lemma 1}}
\label{appendix lemma 1}
We can easily derive
\begin{align}
    \label{lemma 1 1}
  &\mathbb{E}\left[ {\left. {{{\left\| {\sum\limits_{i = 1}^N {I_i^t{\mathbf{g}}_i^t} } \right\|}^2}} \right|{\mathcal{F}^t}} \right] \nonumber \\
   =& \sum\limits_{j = 1}^N {\sum\limits_{i = 1}^N {\mathbb{E}\left[ {\left. {\left\langle {I_i^t{\mathbf{g}}_i^t,I_j^t{\mathbf{g}}_j^t} \right\rangle } \right|{\mathcal{F}^t}} \right]} }  \nonumber \\
   =& \sum\limits_{i = 1}^N {\mathbb{E}\left[ {\left. {\left\langle {I_i^t{\mathbf{g}}_i^t,I_i^t{\mathbf{g}}_i^t} \right\rangle } \right|{\mathcal{F}^t}} \right]}  + \sum\limits_{j = 1}^N {\sum\limits_{i \ne j}^N {\mathbb{E}\left[ {\left. {\left\langle {I_i^t{\mathbf{g}}_i^t,I_j^t{\mathbf{g}}_j^t} \right\rangle } \right|{\mathcal{F}^t}} \right]} }  \nonumber \\
   =& \left( {1 - p} \right)\sum\limits_{i = 1}^N {\mathbb{E}\left[ {\left. {\left\langle {{\mathbf{g}}_i^t,{\mathbf{g}}_i^t} \right\rangle } \right|{\mathcal{F}^t}} \right]} \nonumber\\
   &+ {\left( {1 - p} \right)^2}\sum\limits_{j = 1}^N {\sum\limits_{i = 1}^N {\mathbb{E}\left[ {\left. {\left\langle {{\mathbf{g}}_i^t,{\mathbf{g}}_j^t} \right\rangle } \right|{\mathcal{F}^t}} \right]} }  \nonumber \\
   &- {\left( {1 - p} \right)^2}\sum\limits_{i = 1}^N {\mathbb{E}\left[ {\left. {\left\langle {{\mathbf{g}}_i^t,{\mathbf{g}}_i^t} \right\rangle } \right|{\mathcal{F}^t}} \right]}  \nonumber \\
  \mathop  = \limits^{\left\langle 1 \right\rangle }& \left( {1 - p} \right)\sum\limits_{i = 1}^N {\mathbb{E}\left[ {\left. {{{\left\| {{\mathbf{g}}_i^t} \right\|}^2}} \right|{\mathcal{F}^t}} \right]}  + {{{\left\| {\nabla F\left( {{{\boldsymbol{\theta }}^t}} \right)} \right\|}^2}}\nonumber\\
  &- {\left( {1 - p} \right)^2}\sum\limits_{i = 1}^N {\mathbb{E}\left[ {\left. {{{\left\| {{\mathbf{g}}_i^t} \right\|}^2}} \right|{\mathcal{F}^t}} \right]}  \nonumber \\
   = &\left( {p - {p^2}} \right)\sum\limits_{i = 1}^N {\mathbb{E}\left[ {\left. {{{\left\| {{\mathbf{g}}_i^t} \right\|}^2}} \right|{\mathcal{F}^t}} \right]}  + {{{\left\| {\nabla F\left( {{{\boldsymbol{\theta }}^t}} \right)} \right\|}^2}},
\end{align}
where ${\left\langle 1 \right\rangle }$ is obtained by noting that $\left( {1 - p} \right)\sum\limits_{i = 1}^N {{\mathbf{g}}_i^t}  = \nabla F\left( {{{\boldsymbol{\theta }}^t}} \right)$. In (\ref{lemma 1 1}), we can express
\begin{align}
    \label{lemma 1 2}
  &\sum\limits_{i = 1}^N {\mathbb{E}\left[ {\left. {{{\left\| {{\mathbf{g}}_i^t} \right\|}^2}} \right|{\mathcal{F}^t}} \right]}  \nonumber \\
  \mathop  = \limits^{\left\langle 1 \right\rangle } &\sum\limits_{i = 1}^N {\mathbb{E}\left[ {\left. {{{\left\| {\sum\limits_{k \in {\mathcal{S}_i}} {\frac{1}{{{d_k}\left( {1 - p} \right)}}} \nabla {f_k}\left( {{{\boldsymbol{\theta }}^t}} \right)} \right\|}^2}} \right|{\mathcal{F}^t}} \right]}  \nonumber \\
   =& \frac{1}{{{{\left( {1 - p} \right)}^2}}}\sum\limits_{i = 1}^N {\sum\limits_{{k_1} = 1}^M {\sum\limits_{{k_2} = 1}^M {\frac{{s\left( {i,{k_1}} \right)s\left( {i,{k_2}} \right)}}{{{d_{{k_1}}}{d_{{k_2}}}}} } } }  \nonumber \\
  & \times \left\langle {\nabla {f_{{k_1}}}\left( {{{\boldsymbol{\theta }}^t}} \right),\nabla {f_{{k_2}}}\left( {{{\boldsymbol{\theta }}^t}} \right)} \right\rangle \nonumber\\
   =& \frac{1}{{{{\left( {1 - p} \right)}^2}}}\sum\limits_{{k_1} = 1}^M {\sum\limits_{{k_2} = 1}^M {\frac{1}{{{d_{{k_1}}}{d_{{k_2}}}}}\sum\limits_{i = 1}^N {s\left( {i,{k_1}} \right)s\left( {i,{k_2}} \right) } } }  \nonumber \\
   &\times \left\langle {\nabla {f_{{k_1}}}\left( {{{\boldsymbol{\theta }}^t}} \right),\nabla {f_{{k_2}}}\left( {{{\boldsymbol{\theta }}^t}} \right)} \right\rangle \nonumber\\
   =& \frac{1}{{{{\left( {1 - p} \right)}^2}}}\sum\limits_{{k_1}} {\frac{{{{\left\| {\nabla {f_{{k_1}}}\left( {{{\boldsymbol{\theta }}^t}} \right)} \right\|}^2}}}{{d_{{k_1}}^2}}\sum\limits_{i = 1}^N {s\left( {i,{k_1}} \right)} }  \nonumber \\
   &  + \frac{1}{{{{\left( {1 - p} \right)}^2}}}\sum\limits_{{k_1}} {\sum\limits_{{k_2} \ne {k_1}} {\frac{{\left\langle {\nabla {f_{{k_1}}}\left( {{{\boldsymbol{\theta }}^t}} \right),\nabla {f_{{k_2}}}\left( {{{\boldsymbol{\theta }}^t}} \right)} \right\rangle }}{{{d_{{k_1}}}{d_{{k_2}}}}} } }    \nonumber \\
  & \times \sum\limits_{i = 1}^N {s\left( {i,{k_1}} \right)s\left( {i,{k_2}} \right)} \nonumber\\
  \mathop  = \limits^{\left\langle 2 \right\rangle } &\frac{1}{{{{\left( {1 - p} \right)}^2}}}\sum\limits_k {\frac{{{{\left\| {\nabla {f_k}\left( {{{\boldsymbol{\theta }}^t}} \right)} \right\|}^2}}}{{{d_k}}}} \nonumber\\
  &+ \frac{1}{{{{\left( {1 - p} \right)}^2}}}\sum\limits_{{k_1}} {\sum\limits_{{k_2} \ne {k_1}} {\frac{{\left\langle {\nabla {f_{{k_1}}}\left( {{{\boldsymbol{\theta }}^t}} \right),\nabla {f_{{k_2}}}\left( {{{\boldsymbol{\theta }}^t}} \right)} \right\rangle }}{N}} }  \nonumber \\
   =& \frac{1}{{{{\left( {1 - p} \right)}^2}}}\sum\limits_k {\frac{{{{\left\| {\nabla {f_k}\left( {{{\boldsymbol{\theta }}^t}} \right)} \right\|}^2}}}{{{d_k}}}} - \frac{1}{{{{\left( {1 - p} \right)}^2}}}\sum\limits_k {\frac{{{{\left\| {\nabla {f_k}\left( {{{\boldsymbol{\theta }}^t}} \right)} \right\|}^2}}}{N}}  \nonumber \\
   &+ \frac{1}{{{{\left( {1 - p} \right)}^2}}}\sum\limits_{{k_1}} {\sum\limits_{{k_2}} {\frac{{\left\langle {\nabla {f_{{k_1}}}\left( {{{\boldsymbol{\theta }}^t}} \right),\nabla {f_{{k_2}}}\left( {{{\boldsymbol{\theta }}^t}} \right)} \right\rangle }}{N}} }  \nonumber \\
   = & \frac{1}{{{{\left( {1 - p} \right)}^2}}}\sum\limits_{k = 1}^M {{{\left\| {\nabla {f_k}\left( {{{\boldsymbol{\theta }}^t}} \right)} \right\|}^2}\left( {\frac{1}{{{d_k}}} - \frac{1}{N}} \right)} \nonumber\\
   &+ \frac{1}{{{{\left( {1 - p} \right)}^2}}}\frac{{{{\left\| {\nabla F\left( {{{\boldsymbol{\theta }}^t}} \right)} \right\|}^2}}}{N}\nonumber\\
   = & \frac{1}{{{{\left( {1 - p} \right)}^2}}}\sum\limits_{k = 1}^M {{{\left\| {\nabla {f_k}\left( {{{\boldsymbol{\theta }}^t}} \right) - \frac{1}{M}\nabla F\left( {{{\boldsymbol{\theta }}^t}} \right) + \frac{1}{M}\nabla F\left( {{{\boldsymbol{\theta }}^t}} \right)} \right\|}^2}}  \nonumber \\
   &\times \left( {\frac{1}{{{d_k}}} - \frac{1}{N}} \right) \nonumber \\
  & + \frac{1}{{{{\left( {1 - p} \right)}^2}}}\frac{{{{\left\| {\nabla F\left( {{{\boldsymbol{\theta }}^t}} \right)} \right\|}^2}}}{N} \nonumber \\
  \mathop  \leqslant \limits^{\left\langle 3 \right\rangle } & 2\frac{1}{{{{\left( {1 - p} \right)}^2}}}\sum\limits_{k = 1}^M {{{\left\| {\nabla {f_k}\left( {{{\boldsymbol{\theta }}^t}} \right) - \frac{1}{M}\nabla F\left( {{{\boldsymbol{\theta }}^t}} \right)} \right\|}^2}\left( {\frac{1}{{{d_k}}} - \frac{1}{N}} \right)}  \nonumber \\
  & + 2\frac{1}{{{{\left( {1 - p} \right)}^2}}}\sum\limits_{k = 1}^M {{{\left\| {\frac{1}{M}\nabla F\left( {{{\boldsymbol{\theta }}^t}} \right)} \right\|}^2}\left( {\frac{1}{{{d_k}}} - \frac{1}{N}} \right)}  \nonumber \\
  & + \frac{1}{{{{\left( {1 - p} \right)}^2}}}\frac{{{{\left\| {\nabla F\left( {{{\boldsymbol{\theta }}^t}} \right)} \right\|}^2}}}{N} \nonumber \\
  \mathop  \leqslant \limits^{\left\langle 4 \right\rangle } & \frac{{2{\beta ^2}}}{{{{\left( {1 - p} \right)}^2}}}\sum\limits_{k = 1}^M {\left( {\frac{1}{{{d_k}}} - \frac{1}{N}} \right)}  \nonumber \\
 &  + \left[ {\frac{1}{{N{{\left( {1 - p} \right)}^2}}} + \frac{{2\sum\limits_{k = 1}^M {\left( {\frac{1}{{{d_k}}} - \frac{1}{N}} \right)} }}{{{{\left( {1 - p} \right)}^2}{M^2}}}} \right]{\left\| {\nabla F\left( {{{\boldsymbol{\theta }}^t}} \right)} \right\|^2},
\end{align}
where ${\left\langle 1 \right\rangle }$ is derived from (\ref{encoding}), ${\left\langle 2 \right\rangle }$ holds according to the following equations under the pairwise balanced scheme:
\begin{align}
    \label{PBC}
    \sum\limits_{i = 1}^N {s\left( {i,{k_1}} \right)}  = {d_{{k_1}}},\sum\limits_{i = 1}^N {s\left( {i,{k_1}} \right)s\left( {i,{k_2}} \right)}  = \frac{{{d_{{k_1}}}{d_{{k_2}}}}}{N},k_1 \ne k_2,
\end{align}
${\left\langle 3 \right\rangle }$ holds due to the following basic equality:
\begin{align}
    \label{basic ineq1}
    {\left\| {\sum\limits_{i = 1}^n {{{\mathbf{x}}_i}} } \right\|^2} \leqslant n\sum\limits_{i = 1}^n {{{\left\| {{{\mathbf{x}}_i}} \right\|}^2}} ,\forall {{\mathbf{x}}_i} \in {\mathbb{R}^D},
\end{align}
and ${\left\langle 4 \right\rangle }$ is derived from Assumption~\ref{assp bounded heter}. 
Substituting (\ref{lemma 1 2}) into (\ref{lemma 1 1}) yields (\ref{sum ig}) and completes the proof. 

\section{Proof of Lemma \ref{lemma error}}
\label{appendix lemma error}
Let us define 
\begin{align}
    \label{error1}
    {\mathbf{\tilde e}}_i^t = \gamma {\mathbf{g}}_i^t + {\mathbf{e}}_i^t - \mathcal{C}\left( {\gamma {\mathbf{g}}_i^t + {\mathbf{e}}_i^t} \right), i=1,...,N. 
\end{align}
Based on this definition, we have
\begin{align}
    \label{straggler1}
    {\mathbf{e}}_i^{t + 1} = \left\{ {\begin{array}{*{20}{c}}
  {{\mathbf{\tilde e}}_i^t,}&{{\text{if device }}i{\text{ is a non-straggler}}{\text{,}}} \\ 
  {{\mathbf{e}}_i^t,}&{{\text{otherwise}}{\text{.}}} 
\end{array}} \right.
\end{align}
According to (\ref{straggler1}) and the straggler model in (\ref{straggler model}), we have
\begin{align}
    \label{error 2}
    {\left\| {\sum\limits_{i = 1}^N {{\mathbf{e}}_i^{t + 1}} } \right\|^2} = {\left\| {\sum\limits_{i = 1}^N {I_i^t{\mathbf{\tilde e}}_i^t + \sum\limits_{i = 1}^N {\left( {1 - I_i^t} \right){\mathbf{e}}_i^t} } } \right\|^2} \nonumber\\
    = {\left\| {\sum\limits_{i = 1}^N {I_i^t\left( {{\mathbf{\tilde e}}_i^t - {\mathbf{e}}_i^t} \right) + \sum\limits_{i = 1}^N {{\mathbf{e}}_i^t} } } \right\|^2}.
\end{align}
Taking expectations on both sides of (\ref{error 2}) conditioned on iterations $t=0,...,t-1$, we can obtain
\begin{align}
    \label{exp error 2}
  &\mathbb{E}\left[ {\left. {{{\left\| {\sum\limits_{i = 1}^N {{\mathbf{e}}_i^{t + 1}} } \right\|}^2}} \right|{\mathcal{F}^t}} \right] \nonumber \\
   =& \mathbb{E}\left[ {\left. {\left\langle {\sum\limits_{i = 1}^N {I_i^t\left( {{\mathbf{\tilde e}}_i^t - {\mathbf{e}}_i^t} \right)} ,\sum\limits_{i = 1}^N {I_i^t\left( {{\mathbf{\tilde e}}_i^t - {\mathbf{e}}_i^t} \right)} } \right\rangle } \right|{\mathcal{F}^t}} \right]\nonumber\\
   &+ \mathbb{E}\left[ {\left. {\left\langle {\sum\limits_{i = 1}^N {{\mathbf{e}}_i^t} ,\sum\limits_{i = 1}^N {{\mathbf{e}}_i^t} } \right\rangle } \right|{\mathcal{F}^t}} \right] \nonumber \\
  & + 2\mathbb{E}\left[ {\left. {\left\langle {\sum\limits_{i = 1}^N {{\mathbf{e}}_i^t} ,\sum\limits_{i = 1}^N {I_i^t\left( {{\mathbf{\tilde e}}_i^t - {\mathbf{e}}_i^t} \right)} } \right\rangle } \right|{\mathcal{F}^t}} \right]. 
\end{align}
In (\ref{exp error 2}), the first term can be expressed as 
\begin{align}
    \label{error 2 1st}
  &\mathbb{E}\left[ {\left. {\left\langle {\sum\limits_{i = 1}^N {I_i^t\left( {{\mathbf{\tilde e}}_i^t - {\mathbf{e}}_i^t} \right)} ,\sum\limits_{i = 1}^N {I_i^t\left( {{\mathbf{\tilde e}}_i^t - {\mathbf{e}}_i^t} \right)} } \right\rangle } \right|{\mathcal{F}^t}} \right] \nonumber \\
   =& \sum\limits_i {\mathbb{E}\left[ {\left. {\left\langle {I_i^t\left( {{\mathbf{\tilde e}}_i^t - {\mathbf{e}}_i^t} \right),I_i^t\left( {{\mathbf{\tilde e}}_i^t - {\mathbf{e}}_i^t} \right)} \right\rangle } \right|{\mathcal{F}^t}} \right]}  \nonumber \\
   &+ \sum\limits_{i \ne j} {\mathbb{E}\left[ {\left. {\left\langle {I_i^t\left( {{\mathbf{\tilde e}}_i^t - {\mathbf{e}}_i^t} \right),I_j^t\left( {{\mathbf{\tilde e}}_j^t - {\mathbf{e}}_j^t} \right)} \right\rangle } \right|{\mathcal{F}^t}} \right]}  \nonumber \\
  \mathop  = \limits^{\left\langle 1 \right\rangle } &\left( {1 - p} \right)\sum\limits_i {\mathbb{E}\left[ {\left. {{{\left\| {{\mathbf{\tilde e}}_i^t - {\mathbf{e}}_i^t} \right\|}^2}} \right|{\mathcal{F}^t}} \right]}  \nonumber \\
 &  + {\left( {1 - p} \right)^2}\sum\limits_{i \ne j} {\mathbb{E}\left[ {\left. {\left\langle {\left( {{\mathbf{\tilde e}}_i^t - {\mathbf{e}}_i^t} \right),\left( {{\mathbf{\tilde e}}_j^t - {\mathbf{e}}_j^t} \right)} \right\rangle } \right|{\mathcal{F}^t}} \right]}  \nonumber \\
   =& \left( {p - {p^2}} \right)\sum\limits_i {\mathbb{E}\left[ {\left. {{{\left\| {{\mathbf{\tilde e}}_i^t - {\mathbf{e}}_i^t} \right\|}^2}} \right|{\mathcal{F}^t}} \right]}  \nonumber \\
 &  + {\left( {1 - p} \right)^2}\sum\limits_{i,j} {\mathbb{E}\left[ {\left. {\left\langle {\left( {{\mathbf{\tilde e}}_i^t - {\mathbf{e}}_i^t} \right),\left( {{\mathbf{\tilde e}}_j^t - {\mathbf{e}}_j^t} \right)} \right\rangle } \right|{\mathcal{F}^t}} \right]}  \nonumber \\
   \leqslant & 2\left( {p - {p^2}} \right)\sum\limits_i {\mathbb{E}\left[ {\left. {{{\left\| {{\mathbf{\tilde e}}_i^t} \right\|}^2}} \right|{\mathcal{F}^t}} \right]} \nonumber\\
   & + 2\left( {p - {p^2}} \right)\sum\limits_i {{{\left\| {{\mathbf{e}}_i^t} \right\|}^2}} \nonumber\\
  & + {\left( {1 - p} \right)^2}\mathbb{E}\left[ {\left. {{{\left\| {\sum\limits_i {{\mathbf{\tilde e}}_i^t} } \right\|}^2}} \right|{\mathcal{F}^t}} \right] + {\left( {1 - p} \right)^2}{\left\| {\sum\limits_i {{\mathbf{e}}_i^t} } \right\|^2}\nonumber\\
  & - 2{\left( {1 - p} \right)^2}\mathbb{E}\left[ {\left. {\left\langle {\sum\limits_i {{\mathbf{\tilde e}}_i^t} ,\sum\limits_i {{\mathbf{e}}_i^t} } \right\rangle } \right|{\mathcal{F}^t}} \right],
\end{align}
where ${\left\langle 1 \right\rangle }$ is derived from (\ref{straggler model}). 
Substituting (\ref{error 2 1st}) into (\ref{exp error 2}) yields
\begin{align}
    \label{exp error 3}
 & \mathbb{E}\left[ {\left. {{{\left\| {\sum\limits_{i = 1}^N {{\mathbf{e}}_i^{t + 1}} } \right\|}^2}} \right|{\mathcal{F}^t}} \right] \nonumber  \\
   \mathop  \leqslant \limits^{\left\langle 1 \right\rangle }  & 2\left( {p - {p^2}} \right)\sum\limits_i {\mathbb{E}\left[ {\left. {{{\left\| {{\mathbf{\tilde e}}_i^t} \right\|}^2}} \right|{\mathcal{F}^t}} \right]}  + 2\left( {p - {p^2}} \right)\sum\limits_i {{{\left\| {{\mathbf{e}}_i^t} \right\|}^2}}  \nonumber  \\
   & + {\left( {1 - p} \right)^2}\mathbb{E}\left[ {\left. {{{\left\| {\sum\limits_i {{\mathbf{\tilde e}}_i^t} } \right\|}^2}} \right|{\mathcal{F}^t}} \right] + {p^2}{\left\| {\sum\limits_i {{\mathbf{e}}_i^t} } \right\|^2} \nonumber  \\
  & + 2\left( {p - {p^2}} \right)\mathbb{E}\left[ {\left. {\left\langle {\sum\limits_{i = 1}^N {{\mathbf{e}}_i^t} ,\sum\limits_{i = 1}^N {{\mathbf{\tilde e}}_i^t} } \right\rangle } \right|{\mathcal{F}^t}} \right] \nonumber  \\
   \leqslant &2\left( {p - {p^2}} \right)\sum\limits_i {\mathbb{E}\left[ {\left. {{{\left\| {{\mathbf{\tilde e}}_i^t} \right\|}^2}} \right|{\mathcal{F}^t}} \right]}  + 2\left( {p - {p^2}} \right)\sum\limits_i {{{\left\| {{\mathbf{e}}_i^t} \right\|}^2}}  \nonumber  \\
   &+ \left( {1 - p} \right)\mathbb{E}\left[ {\left. {{{\left\| {\sum\limits_i {{\mathbf{\tilde e}}_i^t} } \right\|}^2}} \right|{\mathcal{F}^t}} \right] + p{\left\| {\sum\limits_i {{\mathbf{e}}_i^t} } \right\|^2},
\end{align}
where ${\left\langle 1 \right\rangle }$ is derived from (\ref{straggler model}).
For the first term in (\ref{exp error 3}), we can derive
\begin{align}
    \label{error 3 1st}
 & \sum\limits_i {\mathbb{E}\left[ {\left. {{{\left\| {{\mathbf{\tilde e}}_i^t} \right\|}^2}} \right|{\mathcal{F}^t}} \right]}  \nonumber \\
  \mathop  = \limits^{\left\langle 1 \right\rangle } & \sum\limits_i {\mathbb{E}\left[ {\left. {{{\left\| {\gamma {\mathbf{g}}_i^t + {\mathbf{e}}_i^t - \mathcal{C}\left( {\gamma {\mathbf{g}}_i^t + {\mathbf{e}}_i^t} \right)} \right\|}^2}} \right|{\mathcal{F}^t}} \right]}  \nonumber \\
  \mathop  \leqslant \limits^{\left\langle 2 \right\rangle } &\delta \sum\limits_i {{{\left\| {\gamma {\mathbf{g}}_i^t + {\mathbf{e}}_i^t} \right\|}^2}}  \nonumber \\
  \mathop  \leqslant \limits^{\left\langle 3 \right\rangle } & 2\delta {\gamma ^2}\sum\limits_i {{{\left\| {{\mathbf{g}}_i^t} \right\|}^2}}  + 2\delta \sum\limits_i {{{\left\| {{\mathbf{e}}_i^t} \right\|}^2}},
\end{align}
where ${\left\langle 1 \right\rangle }$ is obtained from (\ref{error1}), ${\left\langle 2 \right\rangle }$ holds due to Assumption~\ref{assp compre}, and ${\left\langle 3 \right\rangle }$ can be derived based on (\ref{basic ineq1}).

For the third term in (\ref{exp error 3}), we have
\begin{align}
    \label{error 3 4th}
 & \mathbb{E}\left[ {\left. {{{\left\| {\sum\limits_i {{\mathbf{\tilde e}}_i^t} } \right\|}^2}} \right|{\mathcal{F}^t}} \right] \nonumber \\
   =& \mathbb{E}\left[ {\left. {{{\left\| {\sum\limits_i {\gamma {\mathbf{g}}_i^t + {\mathbf{e}}_i^t - \mathcal{C}\left( {\gamma {\mathbf{g}}_i^t + {\mathbf{e}}_i^t} \right)} } \right\|}^2}} \right|{\mathcal{F}^t}} \right] \nonumber \\
  \mathop  \leqslant \limits^{\left\langle 1 \right\rangle } &{q_A}\mathbb{E}\left[ {\left. {{{\left\| {\sum\limits_i {\left( {\gamma {\mathbf{g}}_i^t + {\mathbf{e}}_i^t} \right)} } \right\|}^2}} \right|{\mathcal{F}^t}} \right] \nonumber \\
  \mathop  \leqslant \limits^{\left\langle 2 \right\rangle }& {q_A}\left( {1 + \alpha } \right){\gamma ^2}{\left\| {\sum\limits_i {{\mathbf{g}}_i^t} } \right\|^2} + {q_A}\left( {1 + {\alpha ^{ - 1}}} \right){\left\| {\sum\limits_i {{\mathbf{e}}_i^t} } \right\|^2} \nonumber \\
  \mathop  = \limits^{\left\langle 3 \right\rangle }& \frac{{{q_A}{\gamma ^2}\left( {1 + \alpha } \right)}}{{{{\left( {1 - p} \right)}^2}}}{\left\| {\nabla F\left( {{{\boldsymbol{\theta }}^t}} \right)} \right\|^2} +  {{q_A}\left( {1 + {\alpha ^{ - 1}}} \right)} {\left\| {\sum\limits_i {{\mathbf{e}}_i^t} } \right\|^2},\nonumber\\
  & \forall \alpha>0, 
\end{align}
where ${\left\langle 1 \right\rangle }$ holds due to Assumption~\ref{assp agg error}, ${\left\langle 2 \right\rangle }$ is derived from the following basic inequality:
\begin{align}
    \label{basic ineq2}
    {\left\| {{\mathbf{x}} + {\mathbf{y}}} \right\|^2} \leqslant \left( {1 + \alpha } \right){\left\| {\mathbf{x}} \right\|^2} + \left( {1 + {\alpha ^{ - 1}}} \right){\left\| {\mathbf{y}} \right\|^2}, \nonumber\\
    \forall {\mathbf{x}},{\mathbf{y}} \in {\mathbb{R}^D},\forall \alpha>0, 
\end{align}
and ${\left\langle 3 \right\rangle }$ can be derived from (\ref{encoding}).

Substituting (\ref{error 3 1st}) and (\ref{error 3 4th}) into (\ref{exp error 3}), we have
\begin{align}
    \label{error 4}
 & \mathbb{E}\left[ {\left. {{{\left\| {\sum\limits_{i = 1}^N {{\mathbf{e}}_i^{t + 1}} } \right\|}^2}} \right|{\mathcal{F}^t}} \right] \nonumber  \\
   \leqslant& 4\left( {p - {p^2}} \right)\delta {\gamma ^2}\sum\limits_i {{{\left\| {{\mathbf{g}}_i^t} \right\|}^2}}  + \left( {4\delta  + 2} \right)\left( {p - {p^2}} \right)\sum\limits_i {{{\left\| {{\mathbf{e}}_i^t} \right\|}^2}}  \nonumber  \\
  & + \frac{{{q_A}{\gamma ^2}\left( {1 + \alpha } \right)}}{{1 - p}}{\left\| {\nabla F\left( {{{\boldsymbol{\theta }}^t}} \right)} \right\|^2} \nonumber\\
  &+ \left[ {\left( {1 - p} \right){q_A}\left( {1 + {\alpha ^{ - 1}}} \right) + p} \right]{\left\| {\sum\limits_i {{\mathbf{e}}_i^t} } \right\|^2}.
\end{align}
In (\ref{error 4}), we can bound $\sum\limits_i {{{\left\| {{\mathbf{g}}_i^t} \right\|}^2}}$ by 
\begin{align}
    \label{git}
  &  \sum\limits_i {{{\left\| {{\mathbf{g}}_i^t} \right\|}^2}} \nonumber\\
 \leqslant & \frac{{2{\beta ^2}}}{{{{\left( {1 - p} \right)}^2}}}\vartheta  + \left[ {\frac{1}{{N{{\left( {1 - p} \right)}^2}}} + \frac{{2\vartheta }}{{{{\left( {1 - p} \right)}^2}{M^2}}}} \right]{\left\| {\nabla F\left( {{{\boldsymbol{\theta }}^t}} \right)} \right\|^2},
\end{align}
based on (\ref{def vartheta}) and (\ref{lemma 1 2}). 

For the second term in (\ref{error 4}), we can derive
\begin{align}
    \label{2nd error 4}
 & \mathbb{E}\left[ {\left. {{{\left\| {{\mathbf{e}}_i^{t + 1}} \right\|}^2}} \right|{\mathcal{F}^t}} \right] \nonumber  \\
   =& \left( {1 - p} \right)\mathbb{E}\left[ {\left. {{{\left\| {{\mathbf{\tilde e}}_i^t} \right\|}^2}} \right|{\mathcal{F}^t}} \right] + p{\left\| {{\mathbf{e}}_i^t} \right\|^2} \nonumber  \\
   = &\left( {1 - p} \right)\mathbb{E}\left[ {\left. {{{\left\| {\gamma {\mathbf{g}}_i^t + {\mathbf{e}}_i^t - \mathcal{C}\left( {\gamma {\mathbf{g}}_i^t + {\mathbf{e}}_i^t} \right)} \right\|}^2}} \right|{\mathcal{F}^t}} \right] + p{\left\| {{\mathbf{e}}_i^t} \right\|^2} \nonumber  \\
   \leqslant& \left( {1 - p} \right)\delta {\left\| {\gamma {\mathbf{g}}_i^t + {\mathbf{e}}_i^t} \right\|^2} + p{\left\| {{\mathbf{e}}_i^t} \right\|^2} \nonumber  \\
   \leqslant& 2\left( {1 - p} \right)\delta {\gamma ^2}{\left\| {{\mathbf{g}}_i^t} \right\|^2} + \left[ {2\left( {1 - p} \right)\delta  + p} \right]{\left\| {{\mathbf{e}}_i^t} \right\|^2},
\end{align}
by following similar steps as in (\ref{error 3 1st}). Summing over all devices based on (\ref{git}) and (\ref{2nd error 4}), we have
\begin{align}
    \label{sum 2nd}
  &\sum\limits_i {\mathbb{E}\left[ {\left. {{{\left\| {{\mathbf{e}}_i^{t + 1}} \right\|}^2}} \right|{\mathcal{F}^t}} \right]}  \nonumber  \\
   \leqslant& \frac{{4\delta {\gamma ^2}{\beta ^2}\vartheta }}{{1 - p}} + 2\delta {\gamma ^2}\left[ {\frac{1}{{N\left( {1 - p} \right)}} + \frac{{2\vartheta }}{{\left( {1 - p} \right){M^2}}}} \right]{\left\| {\nabla F\left( {{{\boldsymbol{\theta }}^t}} \right)} \right\|^2} \nonumber\\
   &+ \left[ {2\left( {1 - p} \right)\delta  + p} \right]\sum\limits_i {{{\left\| {{\mathbf{e}}_i^t} \right\|}^2}}.
\end{align}
Taking full expectations on both sides of (\ref{sum 2nd}), we can derive that
\begin{align}
    \label{full exp sum 2nd}
  &\mathbb{E}\left[ {\sum\limits_i {{{\left\| {{\mathbf{e}}_i^{t + 1}} \right\|}^2}} } \right] \nonumber  \\
   \leqslant &\frac{{4\delta {\gamma ^2}{\beta ^2}\vartheta }}{{1 - p}} + \frac{{2\delta {\gamma ^2}}}{{1 - p}}\left( {\frac{1}{N} + \frac{{2\vartheta }}{{{M^2}}}} \right)\mathbb{E}\left[ {{{\left\| {\nabla F\left( {{{\boldsymbol{\theta }}^t}} \right)} \right\|}^2}} \right] \nonumber  \\
  & + \left[ {2\left( {1 - p} \right)\delta  + p} \right]\mathbb{E}\left[ {\sum\limits_i {{{\left\| {{\mathbf{e}}_i^t} \right\|}^2}} } \right] \nonumber  \\
   \leqslant & \frac{{2\delta {\gamma ^2}}}{{1 - p}}\sum\limits_{\tau  = 0}^t {\left\{ {2{\beta ^2}\vartheta  + \left( {\frac{1}{N} + \frac{{2\vartheta }}{{{M^2}}}} \right)\mathbb{E}\left[ {{{\left\| {\nabla F\left( {{{\boldsymbol{\theta }}^{t - \tau }}} \right)} \right\|}^2}} \right]} \right\}} \nonumber\\
  & \times {\left[ {2\left( {1 - p} \right)\delta  + p} \right]^\tau }.
\end{align}
From (\ref{git}) and (\ref{full exp sum 2nd}), we can take full expectations on both sides of (\ref{error 4}) and obtain (\ref{full error 4}). 
\begin{figure*}
\begin{align}
    \label{full error 4}
 & \mathbb{E}\left[ {{{\left\| {\sum\limits_{i = 1}^N {{\mathbf{e}}_i^{t + 1}} } \right\|}^2}} \right] \nonumber  \\
 \leqslant & \frac{{8{\beta ^2}p\delta {\gamma ^2}\vartheta }}{{1 - p}} + \frac{{4p\delta {\gamma ^2}}}{{1 - p}}\left( {\frac{1}{N} + \frac{{2\vartheta }}{{{M^2}}}} \right){\left\| {\nabla F\left( {{{\boldsymbol{\theta }}^t}} \right)} \right\|^2} \nonumber \\
  & + \frac{{\left( {4\delta  + 2} \right)2\delta {\gamma ^2}p}}{{2\left( {1 - p} \right)\delta  + p}}\sum\limits_{\tau  = 0}^t {\left\{ {2{\beta ^2}\vartheta  + \left( {\frac{1}{N} + \frac{{2\vartheta }}{{{M^2}}}} \right)\mathbb{E}\left[ {{{\left\| {\nabla F\left( {{{\boldsymbol{\theta }}^\tau }} \right)} \right\|}^2}} \right]} \right\}} {\left[ {2\left( {1 - p} \right)\delta  + p} \right]^{t - \tau }} \nonumber \\
  & + \frac{{{q_A}{\gamma ^2}\left( {1 + \alpha } \right)}}{{1 - p}}\mathbb{E}\left[ {{{\left\| {\nabla F\left( {{{\boldsymbol{\theta }}^t}} \right)} \right\|}^2}} \right] + \left[ {\left( {1 - p} \right){q_A}\left( {1 + {\alpha ^{ - 1}}} \right) + p} \right]\mathbb{E}\left[ {{{\left\| {\sum\limits_i {{\mathbf{e}}_i^t} } \right\|}^2}} \right] \nonumber\\
   \leqslant & \sum\limits_{l = 0}^t {{{\left[ {\left( {1 - p} \right){q_A}\left( {1 + {\alpha ^{ - 1}}} \right) + p} \right]}^{t - l}}} \left\{ \begin{gathered}
  \frac{{8{\beta ^2}p\delta {\gamma ^2}\vartheta }}{{1 - p}} + \frac{{4p\delta {\gamma ^2}}}{{1 - p}}\left( {\frac{1}{N} + \frac{{2\vartheta }}{{{M^2}}}} \right)\mathbb{E}\left[ {{{\left\| {\nabla F\left( {{{\boldsymbol{\theta }}^l}} \right)} \right\|}^2}} \right] \nonumber \\
   + \frac{{\left( {4\delta  + 2} \right)2\delta {\gamma ^2}p}}{{2\left( {1 - p} \right)\delta  + p}}\sum\limits_{\tau  = 0}^l {\left\{ {2{\beta ^2}\vartheta  + \left( {\frac{1}{N} + \frac{{2\vartheta }}{{{M^2}}}} \right)\mathbb{E}\left[ {{{\left\| {\nabla F\left( {{{\boldsymbol{\theta }}^\tau }} \right)} \right\|}^2}} \right]} \right\}} {\left[ {2\left( {1 - p} \right)\delta  + p} \right]^{l - \tau }} \nonumber \\
   + \frac{{{q_A}{\gamma ^2}\left( {1 + \alpha } \right)}}{{1 - p}}\mathbb{E}\left[ {{{\left\| {\nabla F\left( {{{\boldsymbol{\theta }}^l}} \right)} \right\|}^2}} \right] \nonumber \\ 
\end{gathered}  \right\}\nonumber\\
   = & \sum\limits_{l = 0}^t {{{\left[ {\left( {1 - p} \right){q_A}\left( {1 + {\alpha ^{ - 1}}} \right) + p} \right]}^{t - l}}} \left\{ {\frac{{8{\beta ^2}p\delta {\gamma ^2}\vartheta }}{{1 - p}} + \frac{{\left( {4\delta  + 2} \right)4\delta {\gamma ^2}p{\beta ^2}\vartheta }}{{\left[ {2\left( {1 - p} \right)\delta  + p} \right]}}\frac{{1 - {{\left[ {2\left( {1 - p} \right)\delta  + p} \right]}^{l + 1}}}}{{1 - \left[ {2\left( {1 - p} \right)\delta  + p} \right]}}} \right\} \nonumber \\
   &+ \left[ {\frac{{4p\delta {\gamma ^2}}}{{1 - p}}\left( {\frac{1}{N} + \frac{{2\vartheta }}{{{M^2}}}} \right) + \frac{{{q_A}{\gamma ^2}\left( {1 + \alpha } \right)}}{{1 - p}}} \right]\sum\limits_{l = 0}^t {{{\left[ {\left( {1 - p} \right){q_A}\left( {1 + {\alpha ^{ - 1}}} \right) + p} \right]}^{t - l}}} \mathbb{E}\left[ {{{\left\| {\nabla F\left( {{{\boldsymbol{\theta }}^l}} \right)} \right\|}^2}} \right] \nonumber \\
  &  + \frac{{\frac{{\left( {4\delta  + 2} \right)2\delta {\gamma ^2}p}}{{2\left( {1 - p} \right)\delta  + p}}\left( {\frac{1}{N} + \frac{{2\vartheta }}{{{M^2}}}} \right){{\left[ {\left( {1 - p} \right){q_A}\left( {1 + {\alpha ^{ - 1}}} \right) + p} \right]}^t}}}{{1 - \frac{{\left[ {2\left( {1 - p} \right)\delta  + p} \right]}}{{\left[ {\left( {1 - p} \right){q_A}\left( {1 + {\alpha ^{ - 1}}} \right) + p} \right]}}}}\sum\limits_{\tau  = 0}^t {\frac{{\mathbb{E}\left[ {{{\left\| {\nabla F\left( {{{\boldsymbol{\theta }}^\tau }} \right)} \right\|}^2}} \right]\left\{ {1 - {{\left\{ {\frac{{\left[ {2\left( {1 - p} \right)\delta  + p} \right]}}{{\left[ {\left( {1 - p} \right){q_A}\left( {1 + {\alpha ^{ - 1}}} \right) + p} \right]}}} \right\}}^{t - \tau  + 1}}} \right\}}}{{{{\left[ {\left( {1 - p} \right){q_A}\left( {1 + {\alpha ^{ - 1}}} \right) + p} \right]}^\tau }}}} \nonumber\\
   \leqslant & \frac{1}{{1 - \left[ {\left( {1 - p} \right){q_A}\left( {1 + {\alpha ^{ - 1}}} \right) + p} \right]}}\left\{ {\frac{{8{\beta ^2}p\delta {\gamma ^2}\vartheta }}{{1 - p}} + \frac{{\left( {4\delta  + 2} \right)4\delta {\gamma ^2}p{\beta ^2}\vartheta }}{{2\left( {1 - p} \right)\delta  + p}}\frac{1}{{1 - \left[ {2\left( {1 - p} \right)\delta  + p} \right]}}} \right\} \nonumber \\
  & + \left\{ {\frac{{4p\delta {\gamma ^2}}}{{1 - p}}\left( {\frac{1}{N} + \frac{{2\vartheta }}{{{M^2}}}} \right) + \frac{{{q_A}{\gamma ^2}\left( {1 + \alpha } \right)}}{{1 - p}} + \frac{{\frac{{\left( {4\delta  + 2} \right)2\delta {\gamma ^2}p}}{{2\left( {1 - p} \right)\delta  + p}}\left( {\frac{1}{N} + \frac{{2\vartheta }}{{{M^2}}}} \right)}}{{1 - \frac{{2\left( {1 - p} \right)\delta  + p}}{{\left( {1 - p} \right){q_A}\left( {1 + {\alpha ^{ - 1}}} \right) + p}}}}} \right\} \nonumber \\
   & \times {\left[ {\left( {1 - p} \right){q_A}\left( {1 + {\alpha ^{ - 1}}} \right) + p} \right]^t}\sum\limits_{\tau  = 0}^t {{{\left[ {\left( {1 - p} \right){q_A}\left( {1 + {\alpha ^{ - 1}}} \right) + p} \right]}^{ - \tau }}} \mathbb{E}\left[ {{{\left\| {\nabla F\left( {{{\boldsymbol{\theta }}^\tau }} \right)} \right\|}^2}} \right].
\end{align}
\end{figure*}
Based on (\ref{full error 4}), let us sum over $T+1$ iterations to obtain (\ref{sum full error 4}). 
\begin{figure*}
\begin{align}
    \label{sum full error 4}
  &\sum\limits_{t = 0}^T {\mathbb{E}\left[ {{{\left\| {\sum\limits_{i = 1}^N {{\mathbf{e}}_i^{t + 1}} } \right\|}^2}} \right]}  \nonumber \\
   \leqslant & \frac{{T + 1}}{{1 - \left[ {\left( {1 - p} \right){q_A}\left( {1 + {\alpha ^{ - 1}}} \right) + p} \right]}}\left\{ {\frac{{8{\beta ^2}p\delta {\gamma ^2}\vartheta }}{{1 - p}} + \frac{{\left( {4\delta  + 2} \right)4\delta {\gamma ^2}p{\beta ^2}\vartheta }}{{2\left( {1 - p} \right)\delta  + p}}\frac{1}{{1 - \left[ {2\left( {1 - p} \right)\delta  + p} \right]}}} \right\} \nonumber \\
   + & \left\{ {\frac{{4p\delta {\gamma ^2}}}{{1 - p}}\left( {\frac{1}{N} + \frac{{2\vartheta }}{{{M^2}}}} \right) + \frac{{{q_A}{\gamma ^2}\left( {1 + \alpha } \right)}}{{1 - p}} + \frac{{\frac{{\left( {4\delta  + 2} \right)2\delta {\gamma ^2}p}}{{2\left( {1 - p} \right)\delta  + p}}\left( {\frac{1}{N} + \frac{{2\vartheta }}{{{M^2}}}} \right)}}{{1 - \frac{{2\left( {1 - p} \right)\delta  + p}}{{\left( {1 - p} \right){q_A}\left( {1 + {\alpha ^{ - 1}}} \right) + p}}}}} \right\}\sum\limits_{\tau  = 0}^T {\mathbb{E}\left[ {{{\left\| {\nabla F\left( {{{\boldsymbol{\theta }}^\tau }} \right)} \right\|}^2}} \right]\sum\limits_{t = \tau }^T {{{\left[ {\left( {1 - p} \right){q_A}\left( {1 + {\alpha ^{ - 1}}} \right) + p} \right]}^{t - \tau }}} }  \nonumber \\
   = & \frac{{T + 1}}{{1 - \left[ {\left( {1 - p} \right){q_A}\left( {1 + {\alpha ^{ - 1}}} \right) + p} \right]}}\left\{ {\frac{{8{\beta ^2}p\delta {\gamma ^2}\vartheta }}{{1 - p}} + \frac{{\left( {4\delta  + 2} \right)4\delta {\gamma ^2}p{\beta ^2}\vartheta }}{{2\left( {1 - p} \right)\delta  + p}}\frac{1}{{1 - \left[ {2\left( {1 - p} \right)\delta  + p} \right]}}} \right\} \nonumber \\
    + & \left\{ {\frac{{4p\delta {\gamma ^2}}}{{1 - p}}\left( {\frac{1}{N} + \frac{{2\vartheta }}{{{M^2}}}} \right) + \frac{{{q_A}{\gamma ^2}\left( {1 + \alpha } \right)}}{{1 - p}} + \frac{{\frac{{\left( {4\delta  + 2} \right)2\delta {\gamma ^2}p}}{{2\left( {1 - p} \right)\delta  + p}}\left( {\frac{1}{N} + \frac{{2\vartheta }}{{{M^2}}}} \right)}}{{1 - \frac{{2\left( {1 - p} \right)\delta  + p}}{{\left( {1 - p} \right){q_A}\left( {1 + {\alpha ^{ - 1}}} \right) + p}}}}} \right\}\sum\limits_{\tau  = 0}^T {\mathbb{E}\left[ {{{\left\| {\nabla F\left( {{{\boldsymbol{\theta }}^\tau }} \right)} \right\|}^2}} \right]\frac{{1 - {{\left[ {\left( {1 - p} \right){q_A}\left( {1 + {\alpha ^{ - 1}}} \right) + p} \right]}^{T - \tau  + 1}}}}{{1 - \left[ {\left( {1 - p} \right){q_A}\left( {1 + {\alpha ^{ - 1}}} \right) + p} \right]}}}  \nonumber \\
   \leqslant & \frac{{T + 1}}{{1 - \left[ {\left( {1 - p} \right){q_A}\left( {1 + {\alpha ^{ - 1}}} \right) + p} \right]}}\left\{ {\frac{{8{\beta ^2}p\delta {\gamma ^2}\vartheta }}{{1 - p}} + \frac{{\left( {4\delta  + 2} \right)4\delta {\gamma ^2}p{\beta ^2}\vartheta }}{{2\left( {1 - p} \right)\delta  + p}}\frac{1}{{1 - \left[ {2\left( {1 - p} \right)\delta  + p} \right]}}} \right\} \nonumber \\
   + & \left\{ {\frac{{4p\delta {\gamma ^2}}}{{1 - p}}\left( {\frac{1}{N} + \frac{{2\vartheta }}{{{M^2}}}} \right) + \frac{{{q_A}{\gamma ^2}\left( {1 + \alpha } \right)}}{{1 - p}} + \frac{{\frac{{\left( {4\delta  + 2} \right)2\delta {\gamma ^2}p}}{{2\left( {1 - p} \right)\delta  + p}}\left( {\frac{1}{N} + \frac{{2\vartheta }}{{{M^2}}}} \right)}}{{1 - \frac{{\left[ {2\left( {1 - p} \right)\delta  + p} \right]}}{{\left( {1 - p} \right){q_A}\left( {1 + {\alpha ^{ - 1}}} \right) + p}}}}} \right\}\frac{1}{{1 - \left[ {\left( {1 - p} \right){q_A}\left( {1 + {\alpha ^{ - 1}}} \right) + p} \right]}}\sum\limits_{\tau  = 0}^T {\mathbb{E}\left[ {{{\left\| {\nabla F\left( {{{\boldsymbol{\theta }}^\tau }} \right)} \right\|}^2}} \right]}.
\end{align}
\end{figure*}
By setting $\alpha  = \frac{{2{q_A}}}{{2\delta  + 1 - 2{q_A}}}$ in (\ref{sum full error 4}), under the conditions $\delta  < 0.5$ and ${q_A} < \frac{{2\delta  + 1}}{2}$, we can derive (\ref{bound error}) easily. 
\section{Proof of Theorem~\ref{convergence performance}}
\label{proof th}
    Let us consider the virtual sequence $\left\{ {{{\mathbf{x}}^t},t = 0,...,T} \right\}$, where 
    \begin{align}
        \label{virtual x}
        {{\mathbf{x}}^t} \triangleq {{\boldsymbol{\theta}}^t} - \sum\limits_{i = 1}^N {{\mathbf{e}}_i^t}.
    \end{align}
From (\ref{virtual x}), we can derive 
    \begin{align}
        \label{virtual}
 & {{\mathbf{x}}^{t + 1}} = {{\boldsymbol{\theta }}^{t + 1}} - \sum\limits_{i = 1}^N {{\mathbf{e}}_i^{t + 1}} \nonumber\\
  \mathop  = \limits^{\left\langle 1 \right\rangle }  & {{\boldsymbol{\theta }}^t} - {{{\mathbf{\hat g}}}^t} - \sum\limits_{i = 1}^N {I_i^t{\mathbf{e}}_i^{t + 1}}  - \sum\limits_{i = 1}^N {\left( {1 - I_i^t} \right){\mathbf{e}}_i^{t + 1}}  \nonumber \\
   \mathop  = \limits^{\left\langle 2 \right\rangle }& {{\boldsymbol{\theta }}^t} - \sum\limits_{i = 1}^N {I_i^t\mathcal{C}\left( {\gamma {\mathbf{g}}_i^t + {\mathbf{e}}_i^t} \right)}  - \sum\limits_{i = 1}^N {I_i^t{\mathbf{e}}_i^{t + 1}}  - \sum\limits_{i = 1}^N {\left( {1 - I_i^t} \right){\mathbf{e}}_i^t}  \nonumber \\
   =& {{\boldsymbol{\theta }}^t} - \sum\limits_{i = 1}^N {I_i^t\left[ {\mathcal{C}\left( {\gamma {\mathbf{g}}_i^t + {\mathbf{e}}_i^t} \right) + {\mathbf{e}}_i^{t + 1}} \right]}  - \sum\limits_{i = 1}^N {\left( {1 - I_i^t} \right){\mathbf{e}}_i^t}  \nonumber \\
   \mathop  = \limits^{\left\langle 3 \right\rangle }& {{\boldsymbol{\theta }}^t} - \sum\limits_{i = 1}^N {I_i^t\left( {\gamma {\mathbf{g}}_i^t + {\mathbf{e}}_i^t} \right)}  - \sum\limits_{i = 1}^N {\left( {1 - I_i^t} \right){\mathbf{e}}_i^t}  \nonumber \\
   =& {{\boldsymbol{\theta }}^t} - \sum\limits_{i = 1}^N {{\mathbf{e}}_i^t}  - \gamma \sum\limits_{i = 1}^N {I_i^t{\mathbf{g}}_i^t}  \mathop  = \limits^{\left\langle 4 \right\rangle } {{\mathbf{x}}^t} - \gamma \sum\limits_{i = 1}^N {I_i^t{\mathbf{g}}_i^t},
    \end{align}
where \(\left\langle 1 \right\rangle \) is obtained by substituting (\ref{update model}) into (\ref{virtual}), \(\left\langle 2 \right\rangle \) is derived from (\ref{compress}) and (\ref{server agg}), \(\left\langle 3 \right\rangle \) holds due to (\ref{error update}), and \(\left\langle 4 \right\rangle \) is derived by (\ref{virtual x}).

From Assumption~\ref{assp smooth} and (\ref{virtual}), we have
\begin{align}
    \label{smooth_vir}
  &F\left( {{{\mathbf{x}}^{t + 1}}} \right) \nonumber\\
  \leqslant& F\left( {{{\mathbf{x}}^t}} \right) + \left\langle {\nabla F\left( {{{\mathbf{x}}^t}} \right),{{\mathbf{x}}^{t + 1}} - {{\mathbf{x}}^t}} \right\rangle  + \frac{L}{2}{\left\| {{{\mathbf{x}}^{t + 1}} - {{\mathbf{x}}^t}} \right\|^2} \nonumber \\
   =& F\left( {{{\mathbf{x}}^t}} \right) - \left\langle {\nabla F\left( {{{\mathbf{x}}^t}} \right),\gamma \sum\limits_{i = 1}^N {I_i^t{\mathbf{g}}_i^t} } \right\rangle  + \frac{L}{2}{\left\| {\gamma \sum\limits_{i = 1}^N {I_i^t{\mathbf{g}}_i^t} } \right\|^2}.
\end{align}
Taking expectations on both sides of (\ref{smooth_vir}) conditioned on iterations 
$0,...,t-1$ yields
\begin{align}
    \label{exp smooth}
 & \mathbb{E}\left[ {\left. {F\left( {{{\mathbf{x}}^{t + 1}}} \right)} \right|{\mathcal{F}^t}} \right] - F\left( {{{\mathbf{x}}^t}} \right) \nonumber \\
   \leqslant&  - \gamma \mathbb{E}\left[ {\left. {\left\langle {\nabla F\left( {{{\mathbf{x}}^t}} \right),\sum\limits_{i = 1}^N {I_i^t{\mathbf{g}}_i^t} } \right\rangle } \right|{\mathcal{F}^t}} \right]\nonumber\\
   &+ \frac{{L{\gamma ^2}}}{2}\mathbb{E}\left[ {\left. {{{\left\| {\sum\limits_{i = 1}^N {I_i^t{\mathbf{g}}_i^t} } \right\|}^2}} \right|{\mathcal{F}^t}} \right] \nonumber \\
   = & - \gamma \mathbb{E}\left[ {\left. {\left\langle {\nabla F\left( {{{\boldsymbol{\theta }}^t}} \right),\sum\limits_{i = 1}^N {I_i^t{\mathbf{g}}_i^t} } \right\rangle } \right|{\mathcal{F}^t}} \right]\nonumber\\
   &+ \frac{{L{\gamma ^2}}}{2}\mathbb{E}\left[ {\left. {{{\left\| {\sum\limits_{i = 1}^N {I_i^t{\mathbf{g}}_i^t} } \right\|}^2}} \right|{\mathcal{F}^t}} \right] \nonumber \\
   &+ \gamma \mathbb{E}\left[ {\left. {\left\langle {\nabla F\left( {{{\boldsymbol{\theta }}^t}} \right) - \nabla F\left( {{{\mathbf{x}}^t}} \right),\sum\limits_{i = 1}^N {I_i^t{\mathbf{g}}_i^t} } \right\rangle } \right|{\mathcal{F}^t}} \right].
\end{align}

For the first term in (\ref{exp smooth}), we can derive 
\begin{align}
    \label{1st term}
   &- \gamma \mathbb{E}\left[ {\left. {\left\langle {\nabla F\left( {{{\boldsymbol{\theta }}^t}} \right),\sum\limits_{i = 1}^N {I_i^t{\mathbf{g}}_i^t} } \right\rangle } \right|{\mathcal{F}^t}} \right] \nonumber \\
   \mathop  = \limits^{\left\langle 1 \right\rangle } &  - \gamma \left( {1 - p} \right)\mathbb{E}\left[ {\left. {\left\langle {\nabla F\left( {{{\boldsymbol{\theta }}^t}} \right),\sum\limits_{i = 1}^N {\sum\limits_{k \in {\mathcal{S}_i}} {\frac{{\nabla {f_k}\left( {{{\boldsymbol{\theta }}^t}} \right)}}{{{d_k}\left( {1 - p} \right)}}} } } \right\rangle } \right|{\mathcal{F}^t}} \right] \nonumber \\
   = & - \gamma \mathbb{E}\left[ {\left. {\left\langle {\nabla F\left( {{{\boldsymbol{\theta }}^t}} \right),\sum\limits_{k = 1}^M {\nabla {f_k}\left( {{{\boldsymbol{\theta }}^t}} \right)} } \right\rangle } \right|{\mathcal{F}^t}} \right] \nonumber \\
   =  &- \gamma \mathbb{E}\left[ {\left. {\left\langle {\nabla F\left( {{{\boldsymbol{\theta }}^t}} \right),\nabla F\left( {{{\boldsymbol{\theta }}^t}} \right)} \right\rangle } \right|{\mathcal{F}^t}} \right] \nonumber \\
   = &  - \gamma {\left\| {\nabla F\left( {{{\boldsymbol{\theta }}^t}} \right)} \right\|^2},
\end{align}
where \(\left\langle 1 \right\rangle \) holds due to (\ref{straggler model}).

For the third term in (\ref{exp smooth}), we have
\begin{align}
    \label{3rd term}
  &\gamma \mathbb{E}\left[ {\left. {\left\langle {\nabla F\left( {{{\boldsymbol{\theta }}^t}} \right) - \nabla F\left( {{{\mathbf{x}}^t}} \right),\sum\limits_{i = 1}^N {I_i^t{\mathbf{g}}_i^t} } \right\rangle } \right|{\mathcal{F}^t}} \right] \nonumber \\
   \mathop  \leqslant \limits^{\left\langle 1 \right\rangle } & \frac{{\gamma \rho }}{2}{\left\| {\nabla F\left( {{{\boldsymbol{\theta }}^t}} \right) - \nabla F\left( {{{\mathbf{x}}^t}} \right)} \right\|^2} + \frac{\gamma }{{2\rho }}\mathbb{E}\left[ {\left. {{{\left\| {\sum\limits_{i = 1}^N {I_i^t{\mathbf{g}}_i^t} } \right\|}^2}} \right|{\mathcal{F}^t}} \right] \nonumber \\
   \mathop  \leqslant \limits^{\left\langle 2 \right\rangle }& \frac{{\gamma \rho {L^2}}}{2}{\left\| {{{\boldsymbol{\theta }}^t} - {{\mathbf{x}}^t}} \right\|^2} + \frac{\gamma }{{2\rho }}\mathbb{E}\left[ {\left. {{{\left\| {\sum\limits_{i = 1}^N {I_i^t{\mathbf{g}}_i^t} } \right\|}^2}} \right|{\mathcal{F}^t}} \right] \nonumber \\
   \mathop  = \limits^{\left\langle 3 \right\rangle }& \frac{{\gamma \rho {L^2}}}{2}{\left\| {\sum\limits_{i = 1}^N {{\mathbf{e}}_i^t} } \right\|^2} + \frac{\gamma }{{2\rho }}\mathbb{E}\left[ {\left. {{{\left\| {\sum\limits_{i = 1}^N {I_i^t{\mathbf{g}}_i^t} } \right\|}^2}} \right|{\mathcal{F}^t}} \right], \forall \rho>0, 
\end{align}
where \(\left\langle 1 \right\rangle \) holds due to Young's inequality, \(\left\langle 2 \right\rangle \) can be derived based on Assumption~\ref{assp smooth}, and \(\left\langle 3 \right\rangle \) is obtained by substituting (\ref{virtual x}) into (\ref{3rd term}).

From (\ref{1st term}) and (\ref{3rd term}), we can rewrite (\ref{exp smooth}) as
\begin{align}
    \label{exp smooth2}
  &\mathbb{E}\left[ {\left. {F\left( {{{\mathbf{x}}^{t + 1}}} \right)} \right|{\mathcal{F}^t}} \right] - F\left( {{{\mathbf{x}}^t}} \right) \nonumber \\
   \leqslant&  - \gamma {\left\| {\nabla F\left( {{{\boldsymbol{\theta }}^t}} \right)} \right\|^2} + \left( {\frac{{L{\gamma ^2}}}{2} + \frac{\gamma }{{2\rho }}} \right)\mathbb{E}\left[ {\left. {{{\left\| {\sum\limits_{i = 1}^N {I_i^t{\mathbf{g}}_i^t} } \right\|}^2}} \right|{\mathcal{F}^t}} \right] \nonumber\\
 &  + \frac{{\gamma \rho {L^2}}}{2}{\left\| {\sum\limits_{i = 1}^N {{\mathbf{e}}_i^t} } \right\|^2},\forall \rho  > 0. 
\end{align}
Substituting (\ref{sum ig}) in Lemma~\ref{lemma 1} into (\ref{exp smooth2}), we have
\begin{align}
    \label{exp smooth 3}
 & \mathbb{E}\left[ {\left. {F\left( {{{\mathbf{x}}^{t + 1}}} \right)} \right|{\mathcal{F}^t}} \right] - F\left( {{{\mathbf{x}}^t}} \right) \nonumber \\
   \leqslant& \left\{ {\left( {\frac{{L{\gamma ^2}}}{2} + \frac{\gamma }{{2\rho }}} \right)\left[ {\frac{p}{{\left( {1 - p} \right)N}} + 1 + \frac{{2p\vartheta }}{{\left( {1 - p} \right){M^2}}}} \right] - \gamma } \right\} \nonumber\\
 &\times {\left\| {\nabla F\left( {{{\boldsymbol{\theta }}^t}} \right)} \right\|^2} \nonumber \\
   &+ \frac{{\gamma \rho {L^2}}}{2}{\left\| {\sum\limits_{i = 1}^N {{\mathbf{e}}_i^t} } \right\|^2} + \left( {\frac{{L{\gamma ^2}}}{2} + \frac{\gamma }{{2\rho }}} \right)\frac{{2p{\beta ^2}\vartheta }}{{1 - p}}.
\end{align}
Taking full expectations on both sides of (\ref{exp smooth 3}), we have
\begin{align}
    \label{full exp}
  &\mathbb{E}\left[ {F\left( {{{\mathbf{x}}^{t + 1}}} \right)} \right] - \mathbb{E}\left[ {F\left( {{{\mathbf{x}}^t}} \right)} \right] \nonumber \\
   \leqslant& \left\{ {\left( {\frac{{L{\gamma ^2}}}{2} + \frac{\gamma }{{2\rho }}} \right)\left[ {\frac{p}{{\left( {1 - p} \right)N}} + 1 + \frac{{2p\vartheta }}{{\left( {1 - p} \right){M^2}}}} \right] - \gamma } \right\} \nonumber \\
  & \times \mathbb{E}\left[ {{{\left\| {\nabla F\left( {{{\boldsymbol{\theta }}^t}} \right)} \right\|}^2}} \right] \nonumber \\
  & + \frac{{\gamma \rho {L^2}}}{2}\mathbb{E}\left[ {{{\left\| {\sum\limits_{i = 1}^N {{\mathbf{e}}_i^t} } \right\|}^2}} \right] + \left( {\frac{{L{\gamma ^2}}}{2} + \frac{\gamma }{{2\rho }}} \right)\frac{{2p{\beta ^2}\vartheta }}{{1 - p}}.
\end{align}
Taking the average over $T+1$ iterations based on (\ref{full exp}) yields
\begin{align}
    \label{average exp}
 & \frac{{\mathbb{E}\left[ {F\left( {{{\mathbf{x}}^{T + 1}}} \right)} \right] - \mathbb{E}\left[ {F\left( {{{\mathbf{x}}^0}} \right)} \right]}}{{T + 1}} \nonumber \\
   \leqslant& \left\{ {\left( {\frac{{L{\gamma ^2}}}{2} + \frac{\gamma }{{2\rho }}} \right)\left[ {\frac{p}{{\left( {1 - p} \right)N}} + 1 + \frac{{2p\vartheta }}{{\left( {1 - p} \right){M^2}}}} \right] - \gamma } \right\} \nonumber \\
  & \times \frac{1}{{T + 1}}\sum\limits_{t = 0}^T {\mathbb{E}\left[ {{{\left\| {\nabla F\left( {{{\boldsymbol{\theta }}^t}} \right)} \right\|}^2}} \right]}  \nonumber \\
  & + \frac{{\gamma \rho {L^2}}}{2}\frac{1}{{T + 1}}\sum\limits_{t = 0}^{T + 1} {\mathbb{E}\left[ {{{\left\| {\sum\limits_{i = 1}^N {{\mathbf{e}}_i^t} } \right\|}^2}} \right]}  + \left( {\frac{{L{\gamma ^2}}}{2} + \frac{\gamma }{{2\rho }}} \right)\frac{{2p{\beta ^2}\vartheta }}{{1 - p}}.
\end{align}
Substituting (\ref{bound error}) in Lemma~\ref{lemma error} into (\ref{average exp}) and rearranging the terms, we have (\ref{average1}).
\begin{figure*}
\begin{align}
    \label{average1}
 & \left\{ { - \left( {\frac{{L{\gamma ^2}}}{2} + \frac{\gamma }{{2\rho }}} \right)\left[ {\frac{p}{{\left( {1 - p} \right)N}} + 1 + \frac{{2p\vartheta }}{{\left( {1 - p} \right){M^2}}}} \right] + \gamma  - \frac{{{\gamma ^3}\rho {L^2}}}{2}{\xi _2}} \right\}\frac{1}{{T + 1}}\sum\limits_{t = 0}^T {\mathbb{E}\left[ {{{\left\| {\nabla F\left( {{{\boldsymbol{\theta }}^t}} \right)} \right\|}^2}} \right]}  \nonumber \\
   \leqslant& \frac{{{\gamma ^3}\rho {L^2}}}{2}{\xi _1} + \left( {\frac{{L{\gamma ^2}}}{2} + \frac{\gamma }{{2\rho }}} \right)\frac{{2p{\beta ^2}\vartheta }}{{1 - p}} + \frac{{ - \mathbb{E}\left[ {F\left( {{{\mathbf{x}}^{T + 1}}} \right)} \right] + \mathbb{E}\left[ {F\left( {{{\mathbf{x}}^0}} \right)} \right]}}{{T + 1}}.
\end{align}
\end{figure*}
By setting $\rho  = \frac{1}{{\gamma {\rho _0}}}$ in (\ref{average1}) and using Assumption~\ref{assp lower bound}, we can further bound (\ref{average1}) as
\begin{align}
    \label{average 2}
  &\left( {\gamma  - {\tilde \varepsilon _0}{\gamma ^2}} \right)\frac{1}{{T + 1}}\sum\limits_{t = 0}^T {\mathbb{E}\left[ {{{\left\| {\nabla F\left( {{{\boldsymbol{\theta }}^t}} \right)} \right\|}^2}} \right]}  \nonumber \\
   \leqslant & {\tilde \varepsilon _1}{\gamma ^2} + \frac{{\mathbb{E}\left[ {F\left( {{{\mathbf{x}}^0}} \right)} \right] - {F^*}}}{{T + 1}}, \forall \rho_0>0, 
\end{align}
where
\begin{align}
    \label{ep0}
{\tilde \varepsilon _0} \triangleq \frac{{L + {\rho _0}}}{2}\left[ {\frac{p}{{\left( {1 - p} \right)N}} + 1 + \frac{{2p\vartheta }}{{\left( {1 - p} \right){M^2}}}} \right] + \frac{{{L^2}{\xi _2}}}{{2{\rho _0}}},
\end{align}
and
\begin{align}
    \label{ep1}
    {\tilde \varepsilon _1} \triangleq \frac{{{L^2}}}{{2{\rho _0}}}{\xi _1} + \left( {\frac{L}{2} + \frac{{{\rho _0}}}{2}} \right)\frac{{2p{\beta ^2}\vartheta }}{{1 - p}}.
\end{align}
If we set $\gamma  = \frac{\phi }{{\sqrt {T + 1} }}$, $\phi>0$, for $T > {\left( {{\tilde \varepsilon _0}\phi } \right)^2} - 1$, it holds from (\ref{average 2}) that 
\begin{align}
    \label{average3}
    & \frac{1}{{T + 1}}\sum\limits_{t = 0}^T {\mathbb{E}\left[ {{{\left\| {\nabla F\left( {{{\boldsymbol{\theta }}^t}} \right)} \right\|}^2}} \right]}  \nonumber \\
   \leqslant & \frac{{{\tilde \varepsilon _1}\phi }}{{\sqrt {T + 1}  - {\tilde \varepsilon _0}\phi }} + \frac{{F\left( {{{\mathbf{x}}^0}} \right) - {F^*}}}{{\phi \sqrt {T + 1}  - {\tilde \varepsilon _0}{\phi ^2}}}\nonumber\\
   = & \frac{{{\tilde \varepsilon _1}\phi }}{{\sqrt {T + 1}  - {\tilde \varepsilon _0}\phi }} + \frac{{F\left( {{{\boldsymbol{\theta}}^0}} \right) - {F^*}}}{{\phi \sqrt {T + 1}  - {\tilde \varepsilon _0}{\phi ^2}}}.
\end{align} 
Note that
\begin{align}
    \label{rho1 op}
    {{\tilde \varepsilon }_1} \geqslant {\varepsilon _1} \triangleq \sqrt {\frac{{2{L^2}{\xi _1}p{\beta ^2}\vartheta }}{{1 - p}}}  + \frac{{Lp{\beta ^2}\vartheta }}{{1 - p}},
\end{align}
which can be attained by setting
\begin{align}
    \label{rho0value}
    {\rho _0} = \sqrt {\frac{{{L^2}{\xi _1}\left( {1 - p} \right)}}{{2p{\beta ^2}\vartheta }}}.
\end{align}
Substituting (\ref{rho0value}) into (\ref{ep0}), based on (\ref{average3}) and (\ref{rho1 op}), we can easily derive Theorem~\ref{convergence performance}.

\end{document}